\documentclass[journal]{IEEEtran}
\usepackage{cite}
\usepackage{url}
\usepackage{amsmath}
\allowdisplaybreaks
\usepackage{amssymb}
\usepackage{enumitem}
\usepackage{bm}
\usepackage{algorithm,algorithmic}
\usepackage{mathtools}
\mathtoolsset{showonlyrefs}
\usepackage{subcaption}
\usepackage{graphicx}
\usepackage{refcount}
\usepackage[hidelinks]{hyperref}
\usepackage{color}\definecolor{RED}{rgb}{1,0,0}\definecolor{BLUE}{rgb}{0,0,1}
\DeclareMathOperator*{\minimize}{minimize}

\DeclareMathOperator*\argmin{\text{argmin}}

\def\diag{\text{diag}}

\newcommand{\mathletter}[1]{%
	\expandafter\newcommand\csname b#1\endcsname{\mathbb #1}
	\expandafter\newcommand\csname c#1\endcsname{\mathcal #1}
	\expandafter\newcommand\csname f#1\endcsname{\mathfrak #1}
	\expandafter\newcommand\csname til#1\endcsname{\widetilde #1}
	\expandafter\newcommand\csname ha#1\endcsname{\widehat #1}
	\expandafter\newcommand\csname bf#1\endcsname{\bf #1}
	\expandafter\newcommand\csname s#1\endcsname{\mathsf #1}
}%
\def\mathletters#1{\mathlettersB #1,,}
\def\mathlettersB#1,{\ifx,#1,\else\mathletter #1\expandafter\mathlettersB\fi}
\mathletters{A,B,C,D,E,F,G,H,I,J,K,L,M,N,O,P,Q,R,S,T,U,V,W,X,Y,Z}

\newcommand{\mathletterl}[1]{%
	\expandafter\providecommand\csname v#1\endcsname{\vec{#1}}
}%
\def\mathlettersl#1{\mathlettersC #1,,}
\def\mathlettersC#1,{\ifx,#1,\else\mathletterl #1\expandafter\mathlettersC\fi}
\mathlettersl{a,b,c,d,e,f,g,h,i,j,k,l,m,n,o,p,q,r,s,t,u,v,w,x,y,z}

\def \qed {\hfill \vrule height6pt width 6pt depth 0pt}
\def\bea{\begin{equation}\begin{alignedat}{-1}}
\def\ena{\end{alignedat}\end{equation}}
\def\bee{\begin{equation}}
\def\ene{\end{equation}}

\renewcommand{\vec}[1]{\mathbf{#1}}

\newtheorem{theo}{Theorem}
\newtheorem{lemma}{Lemma}
\newtheorem{assum}{Assumption}

\newtheorem{prop}{Proposition}

\newenvironment{proof}{\begin{IEEEproof}}{\end{IEEEproof}}

\def\T{\mathsf{T}}

\def\bone{{\mathbf{1}}}
\def\bzero{{\mathbf{0}}}

\begin{document}

\title{Fully Asynchronous Distributed Optimization with Linear Convergence in Directed Networks }

\author{Jiaqi~Zhang, Keyou~You
\thanks{J. Zhang and K. You are with the Department of Automation, and BNRist, Tsinghua University, Beijing 100084, China. E-mail: zjq16@mails.tsinghua.edu.cn, youky@tsinghua.edu.cn.}
}
   
\maketitle

\IEEEpeerreviewmaketitle

\begin{abstract}
We consider the distributed optimization problem, the goal of which is to minimize the sum of local objective functions over a directed network. Though it has been widely studied recently, most of the existing algorithms are designed for synchronized or randomly activated implementation, which may create deadlocks in practice. In sharp contrast, we propose a \emph{fully} asynchronous push-pull gradient algorithm (APPG) where each node updates without waiting for any other node by using (possibly stale) information from neighbors. Thus, it is both deadlock-free and robust to any bounded communication delay. Moreover, we construct two novel augmented networks to theoretically evaluate its performance from the worst-case point of view and show that if local functions have Lipschitz-continuous gradients and their sum satisfies the Polyak-\L ojasiewicz condition (convexity is not required), each node of APPG converges to the same optimal solution at a linear rate of $\cO(\lambda^k)$, where $\lambda\in(0,1)$ and the virtual counter $k$ increases by one no matter which node updates. This largely elucidates its linear speedup efficiency and shows its advantage over the synchronous version.  Finally, the performance of APPG is numerically validated via a logistic regression problem on the \emph{Covertype} dataset.  
\end{abstract}

\begin{IEEEkeywords}
  Fully asynchronous, distributed optimization, linear convergence, Polyak-\L ojasiewicz condition
\end{IEEEkeywords}

\section{Introduction}\label{sec1}

As data get larger and more spatially distributed, the distributed optimization over a network of computing nodes (aka. agents or workers) has found numerous applications in multi-agent problems \cite{nedic2009distributed,lin2016distributed,zhang2017distributed} and machine learning \cite{lian2018asynchronous,assran2018stochastic,zhang2019decentralized}. It aims to minimize the sum of local objective functions, i.e., 
\bee\label{original}
\minimize_{\vec x\in\bR^m}\ f(\vec x):=\sum_{i=1}^n f_i(\vec x)
\ene
where $n$ is the number of nodes and the local objective function $f_i$ is only known by node $i$.  Nodes are expected to solve \eqref{original} by only communicating with  neighbors that are defined by the network, see Fig. \ref{fig_graph}. In the empirical risk minimization problem \cite{zhang2019decentralized,tang2018d}, $f_i$ often takes the form
$
f_i(\vx):=\sum_{\xi\in\cD_i} F_i(\vx;\xi)
$
where $\cD_i$ is a local dataset of node $i$, $\vx$ is the model parameter to be optimized, and $F_i(\vx;\xi)$ is the loss of a single sample $\xi$. 

For large-scale optimization problems, it is crucial to design an easily implementable algorithm that is robust to heterogeneous nodes and communication delays. Many existing works focus on synchronous algorithms where all nodes essentially start to compute each iteration simultaneously (c.f. Fig. \ref{fig_sync}) by using a global synchronization scheme that is often not amenable to the distributed setting. 
\begin{figure}[!t]
	\centering
	\begin{subfigure}[c]{0.45\linewidth}
		\centering
        \includegraphics[height=0.51\linewidth]{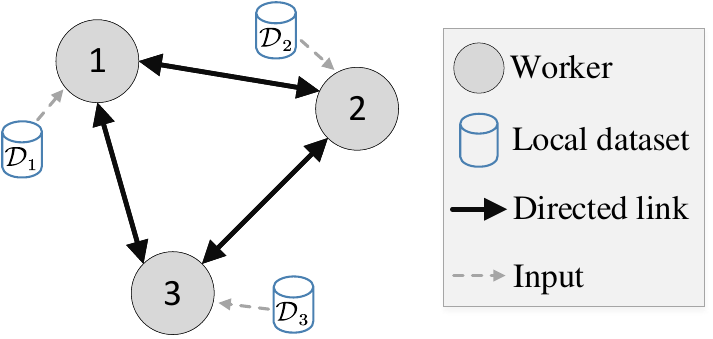}
        \caption{Undirected network}\label{fig_undirected}
	\end{subfigure}
	\hspace{-16pt}
	\begin{subfigure}[c]{0.45\linewidth}
		\centering
        \includegraphics[height=0.51\linewidth]{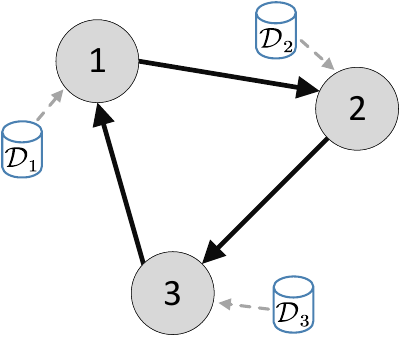}
        \caption{Directed network}\label{fig_directed}
	\end{subfigure}
	\caption{Undirected and directed peer-to-peer networks.}
	\label{fig_graph}
\end{figure}

Asynchronous updates are demonstrated to  perform better than synchronous counterparts \cite{wu2018decentralized,nedic2010asynchronous,bianchi2015coordinate,notarnicola2016asynchronous,xu2018convergence,lian2018asynchronous,zhang2018asyspa,assran2020advances}.  
 A popular one is gossip-based \cite{nedic2010asynchronous,bianchi2015coordinate,notarnicola2016asynchronous,xu2018convergence,lian2018asynchronous}  where a pair of neighbors is randomly selected to concurrently update via information exchange, see Fig. \ref{fig_gossip}. However, this (a) may create deadlocks in practice \cite{tsianos2012consensus,lian2018asynchronous}, especially for networks with many cycles,  (b) is vulnerable to communication delays, and (c) cannot work on directed networks.

To address these issues simultaneously, this work considers the \emph{fully asynchronous} setting (c.f. Fig \ref{fig_async}, \cite{zhang2018asyspa,tian2020achieving,assran2018asynchronous})  over \emph{directed} networks and proposes an asynchronous push-pull gradient  (APPG) algorithm to distributedly solve \eqref{original}. In APPG, a node starts to update without waiting for other nodes by only using locally accessed (possibly stale) information. It does not need any network synchronization, and can tolerate uneven update frequencies  and communication delays among nodes. 

To theoretically evaluate its performance, we develop an augmented network approach and use the machinery of linear matrix inequalities (LMIs) to capture some key quantities via a novel $\lambda$-sequence. If all local functions $f_i$ have Lipschitz continuous gradients and the global objective function $f$ satisfies the Polyak-\L ojasiewicz (PL) condition (no convexity requirement), we prove  from the worst-case point of view that APPG converges linearly to an optimal solution at a rate $\cO(\lambda^k)$ where  $\lambda\in(0,1)$ depends on the asynchrony level and delay bounds, and the virtual counter $k$ increases by one no matter which node updates. Note that the convergence guarantee for gossip-based algorithms is generally given in the stochastic sense.

\begin{figure}[!t]
	\centering
	\begin{subfigure}[c]{\linewidth}
		\centering
        \includegraphics[width=0.8\linewidth]{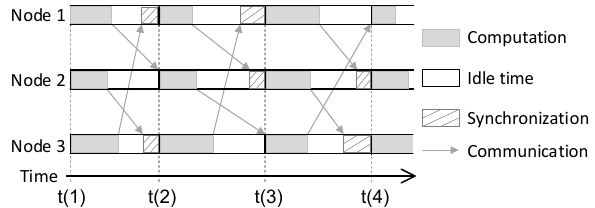}
        \caption{A synchronous algorithm over Fig. \ref{fig_directed}. All nodes start new updates simultaneously.}\label{fig_sync}
	\end{subfigure}
	\\
	\begin{subfigure}[c]{\linewidth}
		\centering
        \includegraphics[width=0.8\linewidth]{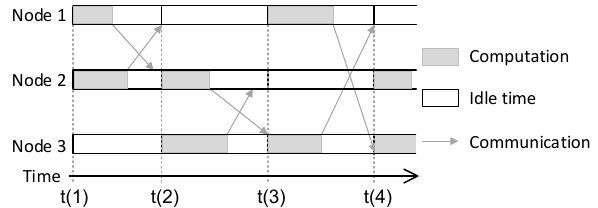}
        \caption{A gossip-based algorithm over Fig. \ref{fig_undirected}. A pair of neighboring nodes are selected to concurrently update via information exchange.}\label{fig_gossip}
	\end{subfigure}
	\\
	\begin{subfigure}[c]{\linewidth}
		\centering
        \includegraphics[width=0.8\linewidth]{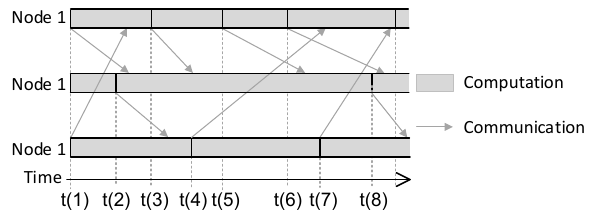}
        \caption{A fully asynchronous algorithm over Fig. \ref{fig_directed}. Every node proceeds without waiting for others.}\label{fig_async}
    \end{subfigure}
	\caption{Synchronous, gossip-based, and fully asynchronous algorithms over the network in Fig. \ref{fig_graph}.}
	\label{fig_algs}
\end{figure}

Then, we implement APPG on a multi-core server with the Message Passing Interface (MPI) to solve a multi-class logistic regression problem over the {\em Covertype} dataset. The result confirms that its empirical convergence rate in running time is faster than its synchronous counterpart, and achieves linear speedup efficiency with respect to the node number. Its robustness to slow computing nodes is also validated, which is essential in heterogeneous large-scale networks.

The rest of the paper is organized as follows. Section \ref{sec0} briefly reviews some related works. Section \ref{sec2} formulates the problem and proposes APPG. We provide the theoretical results for APPG in Section \ref{sec_4}. To prove them, we first develop a time-varying augmented network approach in Section \ref{sec4}, and then establish the LMIs in Section \ref{sec5}. In Section \ref{sec6}, we conduct numerical experiments on APPG. Some concluding remarks are drawn in Section \ref{sec7}.

\section{Related Work}\label{sec0}

The past decade has witnessed an increasing attention on distributed optimization, especially for the design of synchronous algorithms, e.g., DGD \cite{nedic2009distributed}, EXTRA \cite{shi2015extra}, DIGing \cite{nedic2017achieving,qu2017harnessing}, MSDA \cite{scaman2017optimal}, NIDS \cite{li2019decentralized}, ABC \cite{xu2020distributed}, DSGT \cite{zhang2019decentralized} and \cite{yuan2019exact,sun2019distributed,scaman2018optimal,xin2020decentralized,lu2020moniqua}. The row-stochastic matrix \cite{xi2018linear}, epigraph reformulation \cite{xie2018distributed} and the push-sum method \cite{nedic2015distributed,assran2018stochastic,scutari2019distributed} have also been proposed to solve the network unbalancedness issue. For example, Push-DIGing \cite{nedic2017achieving} combines the push-sum method and gradient tracking \cite{xu2018convergence,qu2017harnessing,nedic2017achieving} to achieve linear convergence for strongly convex and Lipschitz smooth functions. However, it involves nonlinear operators per update, which is solved in the push-pull/$\cA\cB$ algorithm \cite{pu2018push,xin2018linear,saadatniaki2018optimization}.

The recent interest has been shift to  asynchronous algorithms. The gossip-based AsynDGM \cite{xu2018convergence} assumes that each pair of neighbors is activated to update with a fixed probability, and has been extended in \cite{lian2018asynchronous} where the activated nodes admit a doubly-stochastic mixing matrix. Ref. \cite{nedic2010asynchronous} considers that each node is activated by a Poisson process. Though they can be applied to directed networks, the coordination between neighbors is indispensable.

Inspired by the seminal works \cite{tsitsiklis1986distributed,li1987asymptotic}, the fully asynchronous setting has emerged as a more scalable and easier implementable alternative. An initial attempt is made in \cite{assran2018asynchronous} to extend the synchronous gradient-push algorithm \cite{nedic2015distributed} to this setting. However, they cannot achieve exact convergence to an optimal solution if nodes have different update frequencies. To address it, Ref. \cite{zhang2018asyspa} proposes a novel adaptive mechanism to dynamically adjust stepsize, which is also adopted in \cite{spiridonoff2020robust} to analyze stochastic algorithms. Nonetheless, the algorithms in \cite{zhang2018asyspa,spiridonoff2020robust} only have sublinear convergence rates even if the objective function is strongly convex and Lipschitz smooth.  Asy-SONATA \cite{tian2020achieving} exploits the perturbed push-sum method with gradient tracking (\cite{xu2018convergence,qu2017harnessing,nedic2017achieving}) in the fully asynchronous setting, and converges (a) linearly for strongly convex and Lipschitz smooth functions, and (b) sublinearly for non-convex problems. The striking differences between \cite{tian2020achieving} and our work include: (a) APPG further converges linearly even under the PL condition, which holds for some important non-convex problems, e.g., the policy optimization for LQR \cite{fazel2018global}. (b) APPG  uses uncoordinated constant stepsizes while Asy-SONATA can only use uncoordinated diminishing stepsizes, which however leads to a sublinear convergence rate. (c) APPG seems easier to understand and implement.

\textbf{Notation}: Throughout this paper, we use the following notation and definitions:
\begin{itemize}[leftmargin=*]
	\item $a,\vec{a},A$, and $\cA$ generally denote respectively a scalar, column vector, matrix, and set. $(\cdot)^\T$ denotes the transpose.
	\item $\bR^n$ and $\bN$ denote the set of $n$-dimensional real numbers and natural numbers, respectively.
	\item $\|\cdot\|_2$ denotes the $l_2$-norm of a vector or matrix. $\|\cdot\|_\sF$ denotes the matrix Frobenius norm. $\cO(\cdot)$ denotes the big-O notation.
	\item $\bone_n$ and $\bzero_n$ denote respectively the $n$-dimensional vector with all ones and all zeros, where the subscript is omitted if the dimension is clear from the context.
	\item $\nabla f(x)$ denotes the gradient of a function $f$ at $x$.
	\item $\va$ is called a stochastic vector if it is nonnegative and $\va^\T\bone=1$. $A$ is called a row-stochastic matrix if $A$ is nonnegative and $A\bone=\bone$. $A$ is column-stochastic if $A^\sT$ is row-stochastic. $A$ is doubly-stochastic if $A$ is both row- and column-stochastic.
	\item $[A]_{ij}$ denotes the element in row $i$ and column $j$ of $A$.
	\item $|\cA|$ denotes the cardinality  of $\cA$.
	\item $\lfloor x\rfloor$ denotes the largest integer less than or equal to $x$.
	\item $[\va_1,\va_2,\dots,\va_n]$ and $[\va_1;\va_2;\dots;\va_n]$ denote the horizontal stack and vertical stack of  $\va_1,\va_2,\dots,\va_n$, respectively.
	\item $\text{Proj}_{\cX}(\vx)$ denotes the projection of $\vx$ onto the set $\cX$.  
\end{itemize}

\section{The APPG}\label{sec2}
\subsection{Problem formulation}\label{sec1a}

We aim to solve \eqref{original} over a directed network. A directed network  (digraph) is denoted by $\cG=(\cV,\cE)$, where $\cV=\{1,2,\cdots,n\}$ is the set of nodes and $\cE\subseteq\cV\times\cV$ is the set of edges. The directed edge $(i,j)\in\cE$ if node  $i$ can directly send information to node $j$. Let $\cN_\text{in}^i=\{j|(j,i)\in\cE\}\cup\{i\}$ denote the set of in-neighbors of node $i$ and  $\cN_\text{out}^i=\{j|(i,j)\in\cE\}\cup\{i\}$ denote the set of out-neighbors of $i$. A path from node $i$ to node $j$ is a sequence of consecutively directed edges from node $i$ to node $j$. Then, $\cG$ is {\em strongly connected} if there exists a directed path between any pair of nodes. For a strongly connected graph $\cG$, the distance between two nodes is the minimum number of edges to connect them via a directed path, and the largest distance $d_g$ is also called its diameter.

For a distributed algorithm over $\cG$, each node $i$ has a local state vector $\vx_i$ and iteratively update it via directed communications with neighbors, the objective of which is to ensure all local states $\vx_i,i\in\cV$ converge to an optimal solution of \eqref{original}. In this work, we make the following assumptions.

\begin{assum}\label{assum}
	\begin{enumerate}
		\renewcommand{\labelenumi}{\rm(\alph{enumi})}
		\item The digraph $\cG$ is strongly connected.
		\item All local functions $f_i$ are $\beta$-Lipschitz smooth, i.e., there exists a  $\beta>0$ such that
		\bee
		\|\nabla f_i(\vx)-\nabla f_i(\vy)\|_2\leq \beta\|\vx-\vy\|_2,\ \forall i\in\cV,\vx,\vy\in\bR^m.
		\ene
		\item The global objective function $f$ has at least one minimizer and satisfies the Polyak-\L ojasiewicz  condition \cite{karimi2016linear} with parameter $\alpha>0$, i.e., $\cX^\star=\argmin_{\vx\in\bR^m} f(\vx)\neq\emptyset$, and
		\begin{equation}\label{pl}
			2\alpha(f(\vx)-f^\star)\leq\|\nabla f(\vx)\|_2^2,\ \forall \vx\in\bR^{m},
			\end{equation}
			where $f^\star=\min_{\vx\in\bR^m} f(\vx)$.
	\end{enumerate}
\end{assum}

Assumptions \ref{assum}(a)-(b) are standard in the distributed smooth optimization over directed networks \cite{nedic2017achieving,pu2018push}. The PL condition in Assumption \ref{assum}(c) is satisfied in some important \emph{non-convex} problems such as the policy optimization for LQR \cite{fazel2018global}. It is strictly weaker than the strongly convex condition that is commonly used to derive the linear convergence of gradient-based methods \cite{nedic2017achieving,tian2020achieving}. Particularly, the strong convexity implies the uniqueness of the minimizer, which is clearly not the case for the PL condition.

\subsection{The APPG}\label{sec3}

The details of APPG are given in Algorithm \ref{alg_APPG}, where we do not introduce any iteration index to emphasize the fact of fully asynchronous implementation.

From the  view of a single node, we illustrate the easy implementation of APPG. Our novel idea lies in the use of local buffers in each node. Particularly, a node just keeps receiving messages from its in-neighbors and storing them to its local buffers until it is activated to compute a new update. Clearly, each buffer may contain zero, one or multiple receptions from the same in-neighbor, which is unavoidable in the fully asynchronous setting.  Another striking feature of APPG is that the node computes a new update by using all messages in the buffers (c.f. \eqref{eq_APPG}), instead of only using the latest reception. This enables APPG to be robust to bounded communication delays, out of sequence issues and there is no need to use any iteration index, which is required by Asy-SONATA \cite{tian2020achieving}.  Interestingly, the local buffers can also be waived since the average and summation operators in \eqref{eq_APPG} can be recursively computed.  Next, the node broadcasts the updated vectors to its out-neighbors, after which empties its buffers. Such a process is repeated until a local stopping criterion (e.g., $\|\vy_i\|_2<\epsilon$) is satisfied.

Different from synchronous or gossip-based algorithms, APPG does not require any global clock or coordination among nodes, and each node does not wait for others for new updates. For example, a node can simply start to compute a new update once it completes the current one. Thus, there is no deadlock problem in APPG.  Moreover,  the local stepsize $\gamma_i$ can be different among nodes.

\begin{algorithm}[t!]
	\makeatletter
	\renewcommand\footnoterule{%
		\kern-3\p@
		\hrule\@width.4\columnwidth
		\kern2.6\p@}
	\makeatother
	\caption{The APPG --- from the view of node $i$}\label{alg_APPG}
	\begin{minipage}{\linewidth}
	\renewcommand{\thempfootnote}{\arabic{mpfootnote}}
	\begin{itemize}[leftmargin=*]
		\item{\bf Initialization:} Each node $i$ selects local stepsize $\gamma_i$, initializes $\vx_i$ as an arbitrary real vector in $\bR^m$, computes $\vg_i=\vec y_i=\nabla f_i(\vx_i)$, and creates local buffers $\cX_i$ and $\cY_i$. Then it broadcasts $\widetilde\vx_i:=\vx_i$ and $\widetilde\vy_i:=\vec y_i/|\cN_\text{out}^i|$ to its out-neighbors.
		\item{{\bf Repeat}}
		\begin{enumerate}[leftmargin=*]
			\renewcommand{\labelenumi}{\theenumi:}
			\item Keep receiving $\widetilde{\vx}_j$ and $\widetilde\vy_j$ from in-neighbors of node $i$ and copy to $\cX_i$ and $\cY_i$ respectively\footnote{The set of in-neighbors includes $i$ itself, i.e., $\widetilde{\vx}_i$ and $\widetilde\vy_i$ are copied to $\cX_i$ and $\cY_i$.}, until node $i$ is activated to update.
			\item Update $\vx_i$ and $\vy_i$ as
			\bea\label{eq_APPG}
			{\vx}_i &\leftarrow \textbf{avg}(\cX_i)\\
			\vg_i^- &\leftarrow \vg_i,\ \vg_i \leftarrow \nabla f_i({\vx}_i)\\
			\vy_i &\leftarrow \textbf{sum}(\cY_i)+\vg_i- \vg_i^-\\
			\widetilde{\vx}_i &\leftarrow {\vx}_i - \gamma_i \vy_i\\
			\ena
			where $\textbf{avg}(\cX_i)$ returns the average\footnote{The average can be weighted, e.g., one may assign higher weights to more recent messages to potentially improve the convergence rate in practice. } of vectors in $\cX_i$, $\textbf{sum}(\cY_i)$ takes the sum of vectors in $\cY_i$.  $\widetilde{\vx}_i$, $\vg_i^-$ and $\vg_i$ are three auxiliary vectors.
			\item Broadcast $\widetilde{\vx}_i$ and $\widetilde\vy_i:=\vy_i/|\cN_\text{out}^i|$ to all out-neighbors of $i$, after which \emph{empty} both $\cX_i$ and $\cY_i$.
		\end{enumerate}
		\item{{\bf Until} a stopping criterion is satisfied. e.g., node $i$ stops if $\|\vy_i\|_2<\epsilon$ for some predefined $\epsilon>0$.}
		\item{{\bf Return} $\vx_i$.}
	\end{itemize}
\end{minipage}
\end{algorithm} 

\subsection{The idea of APPG}\label{sec3.3}

This subsection aims to intuitively explain the linear convergence of the APPG. To this end, we first show how APPG works if nodes are forced to update synchronously. In this case, we can use a global iteration index $k$ to record the update progress.

Let $X(k)$, $Y(k)$ and $\nabla \vf(X(k))$ be the stacked local states and gradients at the $k$-th iteration, i.e.,
\begin{align}
X(k)&=[\vx_1(k),\vx_2(k),\cdots,\vx_n(k)]^\T\in\bR^{n\times m}\\
Y(k)&= [\vy_1(k),\vy_2(k),\cdots,\vy_n(k)]^\T\in\bR^{n\times m}\\
\nabla \vf(X(k))&=[\nabla f_1(\vx_1(k)),\cdots,\nabla f_n(\vx_n(k))]^\T\in\bR^{n\times m}
\end{align}
and $\Gamma=\text{diag}(\gamma_1,\cdots,\gamma_n)$. Then, we obtain that
\begin{subequations}
	\label{ppg}
	\noeqref{ppga,ppgb}
	\begin{align}
	X(k+1) &= A(X(k) - \Gamma Y(k))\label{ppga}\\
	Y(k+1) &= BY(k) + \nabla \vf(X(k+1)) - \nabla \vf(X(k))~~~~~\label{ppgb}
	\end{align}
\end{subequations}
where the row-stochastic matrix $A$ and column-stochastic matrix $B$ result from  \textbf{avg}($\cdot$)  and  \textbf{sum}($\cdot$)  in \eqref{eq_APPG}, respectively. Clearly, \eqref{ppg} reduces to the algorithm in \cite{pu2018push,xin2018linear,saadatniaki2018optimization}, which has been proved to converge linearly for strongly convex and Lipschitz smooth functions. 

The key to the linear convergence of \eqref{ppg} is the introduction of $\vy_i$ to distributedly track the gradient of $f$. Then, we left multiply \eqref{ppgb} with $\bone^\T$, use the column-stochasticity of $B$ and notice $Y(0)=\nabla \vf(X(0))$. This implies that
\bee\label{eq_y}
\bone^\T Y(k)= \bone^\T\nabla \vf(X(k)).
\ene
Now suppose that $X(k)$ and $Y(k)$ have already converged to $X^\infty$ and $Y^\infty$, respectively. It follows from \eqref{ppgb} that $Y^\infty=BY^\infty$. Jointly with \eqref{eq_y}, we obtain that
\begin{equation}\label{eq_yy}
\vy_i^\infty=\pi_i^\sB(\bone^\T\nabla\vf(X^\infty))^\T
\end{equation}
where $\pi^\sB$ is the Perron vector of $B$, i.e. $B\pi^\sB=\pi^\sB$. Moreover, the row-stochasticity of $A$ implies that $X^\infty=\bone (\vx^\infty)^\T$ and
\bee\label{eq_2}
(\bone^\T\nabla\vf(X^\infty))^\sT=(\bone^\T\nabla\vf(\bone (\vx^\infty)^\T))^\sT=\nabla f(\vx^\infty).
\ene
That is, $\vy_i^\infty=\pi_i^\sB \nabla f(\vx^\infty)$. Substituting $X^\infty$ and $Y^\infty$ into \eqref{ppga} implies
$$
X(k+1) = A(X^\infty-\Gamma \pi^\sB \nabla f(\vx^\infty)^\T).
$$
Let $\pi^\sA$ be Perron vector of $A^\T$. We left multiply \eqref{ppga} with $(\pi^\sA)^\T$ and notice that $X^\infty=\bone (\vx^\infty)^\T$. Then,
\bee\label{eq_gd}
\vx(k+1) = \vx^\infty - (\pi^\sA)^\T\Gamma \pi^\sB \nabla f(\vx^\infty)=\vx^\infty -\rho \nabla f(\vx^\infty)
\ene
where $\rho= (\pi^\sA)^\T\Gamma \pi^\sB$. Clearly, \eqref{eq_gd} is a gradient descent update, which converges linearly under Assumption \ref{assum}. It also shows that the limiting point $\vx^\infty$ must be an optimal point $\vx^\star$ and $\vy_i$ converges to 
$\vy_i^\infty=\pi_i^\sB\nabla f(\vx^\infty)=\pi_i^\sB\nabla f(\vx^\star)=0.$ Moreover, the smaller the $\vy_i$, the closer $\vx_i$ to an optimal solution. Therefore, $\vy_i$ can serve as a stopping criterion in Algorithm \ref{alg_APPG}.

In the fully asynchronous setting, $\vy_i$ plays a similar role. However, a theoretical understanding of APPG is much more complicated. The information delays make the key relation \eqref{eq_y} invalid and the uncoordinated updates degrade the tracking performance of $\vy_i$. In essence, APPG is a multi-timescale decision-making problem. To resolve it, we develop an augmented network approach to prove its linear convergence by associating each node with some virtual nodes, under which the asynchronous updates and communications of nodes are transformed into synchronous operations over the augmented network. Moreover, the key technique in  \cite{pu2018push,xin2018linear,saadatniaki2018optimization} cannot be applied since the transformed system is time-varying and lacks their important properties (e.g. the irreducibility of the weighting matrix). To this end, we further design an absolute probability sequence and a $\lambda$-sequence  for the proof.

\section{Linear Convergence of APPG}\label{sec_4}

Two assumptions on the asynchronism and communication delays are needed for the convergence of APPG.

\begin{assum}[Bounded activation time interval]\label{assum3}Let $t_i$ and $t_i^+$ be any two consecutive activation time of node $i$. There exist two positive constants $\underline{\tau}$ and $\bar{\tau}$ such that $0<\underline{\tau}\leq |t_i^+-t_i| \leq \bar{\tau}<\infty$ for all $i\in\cV$.
\end{assum}

Assumption \ref{assum3} is easily satisfied and desirable in practice. In fact, both the lower and upper bounds exist naturally since computing update consumes time and can be finished in finite time. If violated, e.g., some node is broken, then the information from this node can no longer be accessed, and hence it is impossible to find an optimal solution of \eqref{original}.

\begin{assum}[Bounded transmission delays]\label{assum6}For any $(i,j)\in\cE$, the transmission delay from node $i$ to node $j$ is bounded by a constant $\tau>0$.
\end{assum}

 Note that transmission delays can be time-varying, and the parameters $\underline{\tau}$, $\bar{\tau}$ and $\tau$ are not needed for implementing APPG. 

Let $\cT=\{t(k)\}_{k\ge 1}$ be an increasing sequence of updating time of all nodes, i.e., $t\in\cT$ if some node starts to update at time $t$. Denote the state of node $i$ just before time $t(k)$ by $\vec x_i(k)$ and $\vec y_i(k)$. Our main theoretical result is given below.
\begin{theo}\label{theo1}
	Suppose that Assumptions \ref{assum}-\ref{assum6} hold. Let $\bar\gamma=\max_i\gamma_i$ and $\underline\gamma=\min_i\gamma_i$. If $\bar\gamma$ is sufficiently small and $\lambda\in \left(\max\big\{{1-\frac{1-\theta}{2\widetilde{t}}},{1-\frac{\underline\gamma\alpha\theta n}{8b}}\big\},1\right)$, then there exists an optimal solution $\vx^\star\in \cX^\star$ such that
	\bea\label{eq_con}
	\|\vx_i(k)-\vx^\star\|_2&=\cO(\lambda^k),\ 
	\|\vy_i(k)\|_2=\cO(\lambda^k), \forall i\in\cV
	\ena
	where $\alpha$ is given in Assumption \ref{assum}, $b=n(\bar{\tau}+\tau)/\underline{\tau}$. The positive constants $\theta\in(0,1)$ and $\widetilde t$ are given in Lemmas \ref{lemma2} and \ref{lemma3} of Section \ref{sec_4b}, respectively.
\end{theo}

Theorem \ref{theo1} shows that the local decision vector of each node in APPG converges to the same optimal solution at a linear rate. To obtain a deterministic convergence rate, we have to adopt the \emph{worst-case} point of view under the setting that (a) The underlying network is heavily unbalanced, i.e., the differences between the numbers of in-neighbors and out-neighbors of a node is large. (b) Some nodes compute much faster than the others (match the lower and upper bounds of Assumption \ref{assum3}, respectively). (c) Communication delays are based on the upper bound in Assumption \ref{assum6}. Thus, the theoretical rate is expected to be very conservative though the practical performance is empirically much better. 

Roughly speaking, $b$ measures the degree of the asynchronicity and delays in the network. Under Assumptions \ref{assum3} and \ref{assum6}, information sent from a node can be received and used by another node in $b$ steps. $\theta$ and $\widetilde{t}$ reflect the rate of achieving consensus. 

The striking feature is that the virtual counter $k$ in \eqref{eq_con} increases by one no matter which node updates, and hence more nodes generally lead to faster increase of $k$. To some extent, this suggests a linear speedup efficiency \cite{lian2018asynchronous} of APPG, which is confirmed via experiments in Section \ref{sec6}.

\section{The Time-varying Augmented System}\label{sec4}
This section develops our time-varying augmented network approach for the convergence analysis. We first provide a lemma on Assumptions \ref{assum3} and \ref{assum6}.

\begin{lemma}\label{lemma1}
	The following statements hold.
	\begin{enumerate}[label=(\alph*)]
		\item Under Assumption \ref{assum3}, let $b_1=(n-1)\lfloor\bar{\tau}/\underline{\tau}\rfloor+1$. Each node is activated at least once within the time interval $(t(k),t(k+b_1)]$.
		
		\item Under Assumptions \ref{assum3} and \ref{assum6}, let $b_2=n\lfloor\tau/\underline{\tau}\rfloor$ and $b=b_1+b_2$. The information sent from node $i$ at time $t(k)$ can be received by node $j$ before time $t(k+b_2)$ and used to compute an update before time $t(k+b)$ for any $k$ and $(i,j)\in\cE$.
	\end{enumerate}
\end{lemma}
\begin{proof}
	(a) Suppose that node $i$ is not activated during the time interval $(t(p),t(q)],p,q\in\bN$ but is activated at $t(q+1)$. It follows from  Assumption \ref{assum3} that $t(q)-t(p)\leq\bar{\tau}$. Moreover, any other node can be activated at most $\lfloor(t(q)-t(p))/\underline{\tau}\rfloor\leq\lfloor\bar{\tau}/\underline{\tau}\rfloor$ times during the time interval $(t(p),t(q)]$, which implies $q-p\leq (n-1)\lfloor\bar{\tau}/\underline{\tau}\rfloor$. Hence, the first part of the result follows.
	
	(b) Suppose that node $i$ sends information at time $t(p),p\in\bN$ and node $j$ receives it in the time interval $(t(q),t(q+1)],q\in\bN$. It follows from Assumption \ref{assum6} that $t(q)-t(p)\leq\tau$. Moreover,   Assumption \ref{assum3} implies that any node can be activated at most $\lfloor{\tau}/\underline{\tau}\rfloor$ times during the time interval $[t(p),t(q)]$, i.e., $q-p+1\leq n\lfloor {\tau}/\underline{\tau}\ \rfloor$, and hence $q+1\leq p+n\lfloor{\tau}/\underline{\tau}\rfloor$. The result follows by letting $p=k$. Jointly with Lemma \ref{lemma1}(a), the rest of proof follows immediately.
\end{proof}

\subsection{Construction of the augmented digraph}

Let $\cT_i\subseteq\cT$ be the sequence of activation time of node $i$, i.e., $t\in\cT_i$ if node $i$ computes an update at time $t$. Then, it is clear that
\bea\label{reforiter}
&[\vec x_i(k+1), \vec  y_i(k+1), \vec g_i(k+1), \vec  g_i^-(k+1)]\\
&=[\vec x_i(k), \vec  y_i(k), \vec g_i(k), \vec  g_i^-(k)],\ \forall t(k)\notin \cT_i.
\ena

To handle bounded time-varying transmission delays and asynchronicity, we design an augmented system. Firstly, we associate each node $i$ with two types of virtual nodes, and each type has $b$ virtual nodes,  where $b$ is given in Lemma \ref{lemma1}(b). We denote the above two types of virtual nodes by $\{v_{x,i}^{(1)},\cdots,v_{x,i}^{(b)}\}$ and $\{v_{y,i}^{(1)},\cdots,v_{y,i}^{(b)}\}$, respectively. We call the first type virtual nodes $x$-type nodes, which is to deal with the staleness of the state $\widetilde\vx_i,i\in\cV$. The second type with subscript $y$ is called $y$-type nodes, which is to handle the staleness of $\widetilde\vy_i,i\in\cV$. Then, we construct an augmented digraph $\widetilde \cG(k)=(\widetilde{\cV},\widetilde \cE(k))$ to represent the communication topology of all these nodes at time $t(k)$, where $\widetilde{\cV}$ contains $n(2b+1)$ nodes, including  $n$ nodes of $\cG$ and $2nb$ virtual nodes. 

The edge set $\widetilde \cE(k)$  is described as follows. We first note that there is no edge between any $x$-type node and any $y$-type node. For the $x$-type virtual nodes, the edges $(i,v_{x,i}^{(1)}), (v_{x,i}^{(1)},v_{x,i}^{(2)}),\cdots,$ $(v_{x,i}^{(b-2)},v_{x,i}^{(b-1)})$ and $(v_{x,i}^{(b-1)},v_{x,i}^{(b)})$ always include for all $k\in\bN$ and $i\in\cV$ (c.f. Fig. \ref{fig10}). If $(i,j)\in\cE$ in $\cG$ and node $j$ receives $\widetilde\vx_i$ at time $t(k)$, then some of the edges $(v_{x,i}^{(1)},j)$, $(v_{x,i}^{(2)},j),\ldots, (v_{x,i}^{(b)},j)$ and $(i,j)$ are included in $\cE(k)$ (c.f. Fig. \ref{fig20a}), depending on the transmission delay of the received message. If node $j$ receives $\widetilde\vx_i(t-u)$ and $\widetilde\vx_i(t-v)$ at time $t(k)$ for some $u,v> 1$, then $(v_{x,i}^{(u-1)},j),(v_{x,i}^{(v-1)},j)\in\widetilde \cE(k)$. If $u=1$, which means that there is no communication delay, then $(i,j)\in\widetilde{\cE}(k)$.  Fig. \ref{fig20a} illustrates such an augmented graph\footnote{The idea of adding virtual nodes to address asynchronicity or delays was firstly adopted in \cite{nedic2010convergence} to study consensus problems and in \cite{zhang2018asyspa,assran2018asynchronous,tian2020achieving} for distributed optimization. Nevertheless, Refs. \cite{nedic2010convergence,assran2018asynchronous,zhang2018asyspa} use only $x$-type virtual nodes, and hence involve division operators. We further introduce $y$-type nodes to accommodate linear update rules. Ref. \cite{tian2020achieving} associates virtual nodes to original edges rather than original nodes.}.

The topology of the $y$-type virtual nodes is similarly developed with reversed edge  directions (c.f. Fig. \ref{fig30a}), which is the main motivation of using two types of virtual nodes. Firstly, edges $(v_{y,j}^{(1)},j)$, $(v_{y,j}^{(2)},v_{y,j}^{(1)})$,..., and $(v_{y,j}^{(nb)},v_{y,j}^{(nb-1)})$ are always included in $\widetilde{\cE}(k)$ with the reversed directions of $x$-type nodes, see Fig. \ref{fig10}. Secondly, if $(i,j)\in\cE$ in $\cG$ and $k\in\cT_i$, then only one edge in edges $(i,v_j^{(1)})$, $(i,v_j^{(2)}),\ldots, (i,v_j^{(nb)})$ and $(i,j)$ is included in $\cE(k)$, which also depends on the transmission delay of $\widetilde\vy_i$ sent from node $i$ to node $j$. At time $t(k)$, suppose that node $i$ sends $\vy_i(k)$ to node $j$, which is received at $t(k+u)$ for $u>1$, then $(i,v_j^{(u-1)})\in\widetilde{\cE}(k)$ and the delay is $u$. Similarly, if there is no communication delay, i.e., $u=1$, then $(i,j)\in\widetilde{\cE}(k)$.  Fig. \ref{fig30a} illustrates such an augmented graph.

\begin{figure}[!t]
	\centering
	\includegraphics[width=0.7\linewidth]{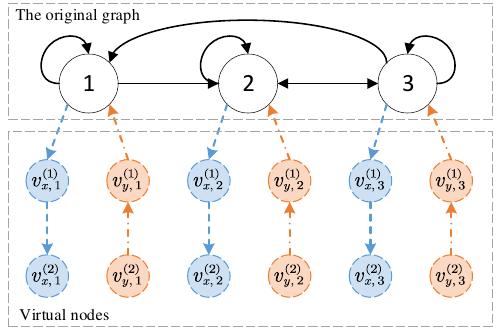}
	\caption{An augmented graph with virtual nodes to address delays of the original graph.}
	\label{fig10}
\end{figure}

\begin{figure}[!t]
	\centering
	\begin{subfigure}[c]{0.5\linewidth}
		\centering
        \includegraphics[width=\linewidth]{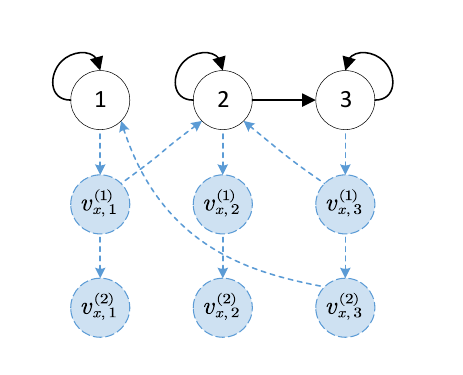}
        \caption{}\label{fig20a}
	\end{subfigure}
	\begin{subfigure}[c]{0.41\linewidth}
		\centering
        \includegraphics[width=\linewidth]{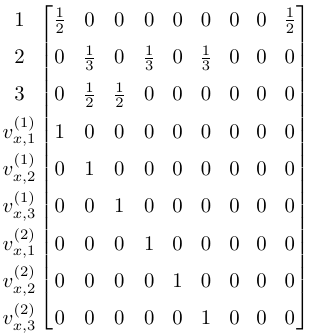}
        \caption{}\label{fig20b}
	\end{subfigure}
	\caption{(a) The topology of the $x$-type virtual nodes in the augmented graph at some time $t(k)$. Here node 1 uses the 2-steps delayed information $\widetilde\vx_3(k-3)$ and the latest information $\widetilde\vx_1(k-1)$ to compute $\vx_1(k)$, and hence $(v_{x,3}^{(2)},1)\in\widetilde\cE(k)$. Node 2 uses $\widetilde\vx_2(k-1)$ and the 1-step delayed information $\widetilde\vx_1(k-2)$ and $\widetilde\vx_3(k-2)$ to compute $\vx_2(k)$. Node 3 use the latest information $\widetilde\vx_2(k-1)$ and $\widetilde\vx_3(k-1)$ to compute $\vx_3(k)$. (b) The corresponding row-stochastic matrix $\widetilde{A}(k)$ in \eqref{eq1_sec2}.}
	\label{fig20}
\end{figure}

\begin{figure}[!t]
	\centering
	\begin{subfigure}[c]{0.5\linewidth}
		\centering
        \includegraphics[width=\linewidth]{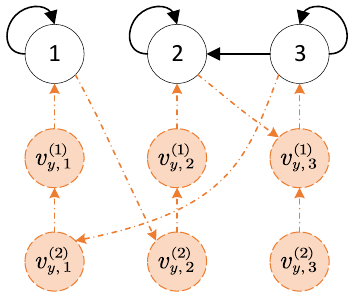}
        \caption{}\label{fig30a}
	\end{subfigure}
	\begin{subfigure}[c]{0.41\linewidth}
		\centering
        \includegraphics[width=\linewidth]{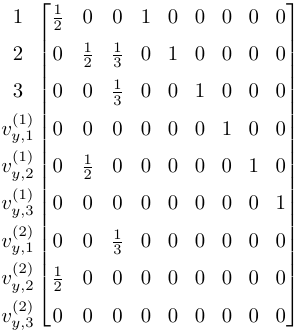}
        \caption{}\label{fig30b}
	\end{subfigure}
	\caption{(a) The topology of the $y$-type virtual nodes in the augmented graph at some time $t(k)$, which represents that node 1 sends $\widetilde\vy_1(k)$ to node 2 and node 2 use it at $t(k+3)$ to compute $\vy_2(k+3)$. Other edges can be interpreted similarly. (b) The corresponding column-stochastic matrix $\widetilde{B}(k)$ in \eqref{eq1_sec2}.}
	\label{fig30}
\end{figure}

We provide a simple example to visualize  the augmented graph approach. Consider that node $i$ sends $\widetilde\vx_i(k)$ and $\widetilde\vy_i(k)$ to node $j$ at time $t(k)$, and node $j$ receives it at time $t(k+2)$, i.e., the delay is $1$. In the augmented graph, this can be viewed as node $i$ directly sends $\widetilde\vx_i(k)$ to the virtual node $v_{x,i}^{(1)}$, and sends $\widetilde\vy_i(k)$ to the virtual node $v_{y,j}^{(1)}$ at time $t(k)$. Nodes $v_{x,i}^{(1)}$ and $v_{y,j}^{(1)}$ respectively receive $\widetilde\vx_i(k)$ and $\widetilde\vy_i(k)$ at time $t(k+1)$, and immediately send them to node $j$ at time $t(k+1)$. Finally, node $j$ receives $\widetilde\vx_i(k)$ and $\widetilde\vy_i(k)$ at time $t(k+2)$. Clearly, all non-virtual nodes in $\widetilde \cG$ receive the same information as that in $\cG$ and hence their updates appear to be synchronous and delay-free.

Under the time-varying  augmented digraph, we are able to rewrite the APPG in a compact form.

\subsection{A compact form of the APPG over the augmented digraph}\label{sec_4b}

Let $\vx_i^{(u)}(k)$ and $\vy_i^{(u)}(k)$ denote the states of virtual node $v_{x,i}^{(u)}$ and $v_{y,i}^{(u)}$ just after time $t(k)$, and $\widetilde n=n(2b+1)$. Then, APPG can be rewritten in a compact form over $\widetilde \cG(k)$,
\bea\label{eq1_sec2}
\widetilde{X}(k+1)&=\widetilde A(k)(\widetilde{X}(k)-\Gamma I_k^a\widetilde Y(k)),\\
\widetilde{Y}(k+1)& = \widetilde B(k)\widetilde{Y}(k)+ I_k^a(\nabla(k+1)-\nabla(k))\\
& = \widetilde B(k)\widetilde{Y}(k)+ \nabla(k+1)-\nabla(k)
\ena
where
\bea\label{eq3_sec2}
\widetilde{X}(k)&=[X(k); X^{(1)}(k);\cdots; X^{(b)}(k)]\in\bR^{\widetilde n\times m}\\  X^{(u)}(k)&=[\vx_1^{(u)}(k),\cdots,\vx_n^{(u)}(k)]^\T\\
\widetilde{Y}(k)&=[Y(k); Y^{(1)}(k);\cdots; Y^{(b)}(k)]\in\bR^{\widetilde n\times m}\\  Y^{(u)}(k)&=[\vy_1^{(u)}(k),\cdots,\vy_n^{(u)}(k)]^\T\\
\nabla(k)&=[\nabla \vf(X(k));\bzero_{(\widetilde{n}-n)\times m}],
\ena
the initial condition is
$
\widetilde{X}(0)=[X(0);\bzero_{(\widetilde{n}-n)\times m}],  \widetilde{Y}(0)=\nabla(0)
$
, and $\widetilde A(k),\widetilde B(k),I_k^a\in\bR^{\widetilde n\times \widetilde n}$ are \bea
&[\widetilde A(k)]_{ij}=\\
&\left\{\begin{array}{ll}
	\frac{1}{|\cX_i(k)|}, &\begin{tabular}{@{}l@{}}\text{if $i,v\in\cV$, $j=nu+v$, $t(k+1)\in\cT_i$,}\\ \text{and node $i$ receives $\vx_v(k-u)$ at $t(k+1)$}\end{tabular}\\
	1,                            & \text{if $i\in\cV$, $t(k+1)\notin\cT_i$ and $j=i$}                                                                     \\
	1,                            & \text{if $i\notin\cV$ and $j=i-n$}                                                                                     \\
	0,                            & \text{otherwise,}
\end{array}\right.\\
\ena
\bea
&[\widetilde B(k)]_{ji}=\\
&\left\{\begin{array}{ll}
	\frac{1}{|\cN_\text{out}^i|}, & \begin{tabular}{@{}l@{}}\text{if $i,v\in\cV$, $j=nu+v$, $t(k+1)\in\cT_i$,}\\ \text{and node $v$ receives $\vx_i(k)$ at $t(k+u)$}\end{tabular} \\
	1,                            & \text{if $i\in\cV$, $t(k+1)\notin\cT_i$ and $j=i$}                                                                        \\
	1,                            & \text{if $i\notin\cV$ and $j=i-n$}                                                                                        \\
	0,                            & \text{otherwise,}
\end{array}\right.\\
\ena
\bea
&[I_k^a]_{ij}=\left\{\begin{array}{ll}
	1, & \text{if $i=j$, $i\in\cV$, and $t(k+1)\in\cT_i$} \\
	0, & \text{otherwise,}
\end{array}\right.
\ena
where $|\cX_i(k)|$ is the number of elements in the buffer $\cX_i$ at time $t(k+1)$.

An example of $\tilA(k)$ and $\tilB(k)$ is illustrated in Fig. \ref{fig20b} and Fig. \ref{fig30b}, respectively. $I_k^a$ is a diagonal matrix with its $i$-th diagonal element be 1 if node $i$ is activated at time $t(k+1)$. The third equality in \eqref{eq1_sec2} follows from $\nabla(k+1)=\nabla(k)$ for any $i\in\{i|[I_k^a]_{ii}=0\}$. An important fact is that $\widetilde A(k)$ is a row-stochastic matrix and $\widetilde B(k)$ is a column-stochastic matrix by the use of two types virtual nodes. Moreover,
\bee\label{eq4_sec2}
\bone_{\widetilde n}^\T\tilY(k)=\bone_{\widetilde n}^\T\nabla(k)=\bone_{n}^\T\nabla\vf(X(k))
\ene
which is obtained by left multiplying the second equality of \eqref{eq1_sec2} with $\bone_{\widetilde n}^\T$.

Note that \eqref{eq1_sec2} generates the same sequence of the states $\vx_i$ and $\vy_i$ as that of APPG. Hence, it is sufficient to study the convergence of $\tilX(k)$ and $\tilY(k)$ in \eqref{eq1_sec2}. To this end, we define
\bea\label{eq2_sec2}
\Phi_t^\sA (k)&=\widetilde A(k+t-1)\widetilde A(k+t-2)\cdots\widetilde A(k+1)\widetilde A(k)\\
\Phi_t^\sB (k)&=\widetilde B(k+t-1)\widetilde B(k+t-2)\cdots\widetilde B(k+1)\widetilde B(k).
\ena
where $k,t\in\bN$, and we adopt the convention that $\Phi_0^\sA (k)=\Phi_0^\sB (k)=I$ and $\Phi_t^\sA (k)=\Phi_t^\sB (k)=0$ for any $k\in\bN$ and $t<0$.

The following lemma states that $\Phi_t^\sA (k)$ and $\Phi_t^\sB (k)$ linearly converge to rank-one matrices.
\begin{lemma}\label{lemma2}
	Under Assumptions \ref{assum}, \ref{assum3} and \ref{assum6}, the following statements are in force.
	\begin{enumerate}[label=(\alph*)]
		\item There exist two stochastic vectors $\phi_{t}^\sA (k)$ and $\phi_{t}^\sB (k)$ such that
		\bee
		\|\Phi_t^\sA (k)- \bone\phi_{t}^\sA (k)^\T\|_\sF\leq 2\rho^{t},\ \|\Phi_t^\sB (k)- \phi_{t}^\sB (k)\bone^\T\|_\sF\leq2\rho^{t}
		\ene 
		for all $k,t\in\bN$, where
		\bee
		\rho=\left(1-\theta\right)^{\frac{1}{d_g b}}<1,\  \theta\geq\left(\frac{1}{\widetilde  n}\right)^{d_g b}\in(0,1),
		\ene
		$b$ is defined in Lemma \ref{lemma1}(b), $d_g$ is the diameter of $\cG$ and $\widetilde{n}=n(2b+1)$.
		\item $\sum_{j=1}^{n}[\Phi_t^\sB(0)]_{ij}\geq n\theta,\ \forall i\in\cV, t\in\bN.$
	\end{enumerate}
\end{lemma}
\begin{proof}
	In view of Lemma \ref{lemma1}, both $\Phi_b^\sA (k)$ and $\Phi_b^\sB (k)$ are primitive for all $k$. Moreover, the minimum value of nonzero elements of $\tilA(k)$ and $\tilB(k)$ is greater than $1/\widetilde{n}$. Then, the proof of the first part is similar with that of Lemma 5 in \cite{nedic2010convergence}.
	
	To prove (b), two cases are separately studied. If $t< d_gb$, then $[\Phi_{t}^\sB(0)]_{ii}\geq[\widetilde B(t-1)]_{ii}[\widetilde B(t-2)]_{ii}\cdots[\widetilde B(0)]_{ii}\geq (1/\widetilde{n})^{d_gb-1}=  \widetilde{n}\theta\geq n\theta$, and hence the result is obtained.
	
	If $t\geq d_gb$, it follows from a similar argument with the Lemma 2(b) in \cite{nedic2010convergence} that
	$[\Phi_{d_gb}(k)]_{ij}\geq \theta$ for all $i\in\cV$ and $j\in\widetilde\cV$. Then,
	\bea
	&[\Phi_{d_gb+1}(k-1)]_{ij}\\
	&=\sum_{u=1}^{\widetilde n}[\Phi_{d_gb}(k)]_{iu}[\widetilde B(k)]_{uj} \geq \theta\sum_{u=1}^{\widetilde n}[\widetilde B(k)]_{uj}\geq \theta.
	\ena
	where the last inequality follows from the column-stochasticity of $\widetilde B(k)$. The desired result is obtained by induction.
\end{proof}

The following lemma is a direct result of Lemma \ref{lemma2}, which specifies $\mu$ and $\widetilde{t}$ in Theorem \ref{theo1}.
\begin{lemma}\label{lemma3}
	Under assumptions of  Lemma \ref{lemma2}, let $\mu=\frac{\theta^2}{2\widetilde n}<1$ and $\widetilde t\in\bN$ be a number such that $2\rho^{\widetilde t}\leq\mu$, then for all $k\in\bN$,
	\bea
	\|\Phi_{\widetilde t}^\sA (k)- \bone\phi_{\widetilde t}^\sA (k)^\T\|_\sF\leq\mu,\ \|\Phi_{\widetilde t}^\sB (k)- \phi_{\widetilde t}^\sB (k)\bone^\T\|_\sF\leq\mu
	\ena
\end{lemma}

Finally, we introduce the \emph{absolute probability sequence}  \cite{touri2012product}.
\begin{lemma}[Theorem 4.2 in \cite{touri2012product}]\label{lemma7}
	For a sequence of row-stochastic matrices $\{A(k)\}$, there exists a sequence of stochastic vectors $\{\pi(k)\}$ satisfying
	\bee\label{abseq}
	\pi(k+1)^\T A(k)=\pi(k)^\T,\ \forall k\in\bN.
	\ene
	$\{\pi(k)\}$ is called an absolute probability sequence of $\{A(k)\}$.
\end{lemma}

In the sequel, we use $\pi(k)\in\bR^{\widetilde n}$ to denote an absolute probability sequence of $\tilA(k)$, which implies that
$\pi(k+t)^\T\Phi_t^\sA(k)=\pi(k)^\T,\forall k,t\in\bN$.

\section{Proof of Theorem \ref{theo1} via LMIs}\label{sec5}

Under the augmented time-varying digraph, we are ready to prove Theorem \ref{theo1}.

\subsection{Outline of the proof}
As \cite{nedic2017achieving}, for a nonnegative sequence $\{p(k)\}$, we define that
\bee\label{eq1_sec4}
p^{\lambda,k} =\sup_{t\in\bN,t\leq k}\frac{p(t)}{\lambda^t}~ \text{and} ~\lambda\in(0,1).
\ene
We call $\{p^{\lambda,k}\}$ the $\lambda$-sequence of $p(k)$. Clearly, if $p^{\lambda,k}$ is uniformly bounded by some constant $c$, then $p(k)\leq c\lambda^k$ for all $k$. Our method to prove Theorem 1 is to show the boundedness of $p^{\lambda,k},k\in\bN$ in \eqref{eq1_sec4} for some nonnegative sequences $p(k)$.

Under Assumptions \ref{assum}, \ref{assum3} and \ref{assum6} and \eqref{eq1_sec2}, the proof of Theorem 1 relies on four lemmas, whose proofs  are given in next subsections.

\begin{lemma}\label{lemma_m1} 
Let 
\bee
Q(k)=I_{\widetilde n}-\bone_{\widetilde n}\pi(k-1)^\T,\ \|\tilX(k)\|_\sQ=\|Q(k)\tilX(k)\|_\sF,
\ene
and $\|\tilX\|_\sQ^{\lambda,k}$ be the $\lambda$-sequence of $\|\tilX(k)\|_\sQ$. If $\lambda^{\widetilde t}\geq\frac{2\theta^2}{n}$, where $\widetilde t$ is defined in Lemma \ref{lemma3}, then
\bea\label{eq_s1}
\|\widetilde{X}\|_\sQ^{\lambda,k}\leq \frac{4\widetilde{n}\widetilde{t}\bar{\gamma}}{\theta^2}\|\widetilde{Y}\|_\sF^{\lambda,k}+c_1
\ena
where $\|\widetilde{Y}\|_\sF^{\lambda,k}$ is the $\lambda$-sequence of $\|\widetilde{Y}(k)\|_\sF$ and $c_1$ is a constant.\hfill$\square$
\end{lemma}

$\|\tilX(k)\|_\sQ$ in Lemma \ref{lemma_m1}  is a weighted difference among nodes' states $\vx_i(k)$,  and \eqref{eq_s1} bounds it by the gradient estimate $\|\widetilde{Y}(k)\|_\sF$.

\begin{lemma}\label{lemma_m2} 
Let
\bea\label{eq4_s2}
\vec v(k+1)&=\widetilde{B}(k)\vec v(k),\ &\vec v(0)=[\bone_{n};\bzero_{\widetilde n-n}],\\ V(k)&=\diag(\vv(k)),\ &Y_\sV(k)= V(k)^{\dag}\widetilde{Y}(k)
\ena
where $V(k)^{\dag}$ is the pseudo inverse of $V(k)$, i.e.,
\bea
[V(k)^{\dag}]_{ij}=\left\{\begin{array}{ll}
	1/[V(k)]_{ii}, & \text{if $i=j$ and $[V(k)]_{ii}>0$,} \\
	0,             & \text{otherwise.}
\end{array}\right.
\ena
Let $\tilI(k)=V(k)V(k)^{\dag}$, $\widetilde\bone(k)=\tilI(k)\bone_{\widetilde n}$ and
\bee\label{eq1_c2}
S(k)=\tilI(k)-\frac{1}{n}\widetilde \bone(k)\vv(k)^\T,\ \|Y_\sV(k)\|_\sS=\|S(k)Y_\sV(k)\|_\sF
\ene
Define the corresponding $\lambda$-sequence $\|Y_\sV\|_\sS^{\lambda,k}$ of $\{\|Y_\sV(k)\|_\sS\}$. If $\lambda^{\widetilde t}>\frac{2\theta}{1+\theta}$, where $\widetilde t$ and $\mu$ are defined in Lemma \ref{lemma3}, then
\bee\label{eq_s2}
\|Y_\sV\|_\sS^{\lambda,k}\leq
\frac{8\beta\widetilde n\sqrt{\widetilde n}\widetilde t}{\theta^2 n}\left(\sqrt{n}\|\widetilde{X}\|_\sQ^{\lambda,k}+\bar{\gamma}\|\widetilde{Y}\|_\sF^{\lambda,k}\right)+c_2
\ene
where $\beta$ is given in Assumption \ref{assum}, $\theta,\widetilde t,\widetilde{n}$ are defined in Lemmas \ref{lemma2} and \ref{lemma3}, and $c_2$ is given by \eqref{eq_c2}.\hfill$\square$
\end{lemma}

Similarly, $\|Y_\sV(k)\|_\sS$ measures the differences between the weighted gradient estimates of different nodes, which is bounded by $\|\tilX(k)\|_\sQ$ and $\|\widetilde{Y}(k)\|_\sF$.

\begin{lemma}\label{lemma_m3}
Let
\bee\label{eq9_s3}
\vec x_{\pi}(k)=\vec \pi(k)^\T\tilX(k)
\ene
Define $\widetilde f^{\lambda,k}$ be the $\lambda$-sequence of $\{\sqrt{f(\vec x_{\pi}(k))-f^\star}\}$. If $\bar{\gamma}\leq \frac{1}{4nb\beta}$ and $\lambda^b\geq 1-\frac{1}{8}\alpha \underline\gamma\theta n$. Then,
\bee\label{eq_s3}
\widetilde f^{\lambda,k}\leq\frac{16b}{\alpha\theta\sqrt{\bar{\gamma} \theta n}}\left({3n\beta}\|\tilX\|_\sQ^{\lambda,k}+\|Y_\sV\|_\sS^{\lambda,k}\right)+c_3
\ene
where $\alpha,\beta$ are given in Assumption \ref{assum}, $b$ is defined in Lemma \ref{lemma1}, $\theta$ is defined in Lemma \ref{lemma2}, and $c_3$  is a constant.\hfill$\square$
\end{lemma}

Intuitively, $\vec x_{\pi}(k)$ is a weighted average of $\vx_i(k),i\in\cV$, and $f(\vec x_{\pi}(k))-f^\star$ is the optimality gap.  Eq. \eqref{eq9_s3} shows that the square root of the optimality gap can be bounded by $\|\tilX\|_\sQ^{\lambda,k}$ and $\|Y_\sV\|_\sS^{\lambda,k}$.

\begin{lemma}\label{lemma_m4} 
With the above-defined $\|\tilX\|_\sQ^{\lambda,k}$, $\|Y_\sV\|_\sS^{\lambda,k}$ and $\widetilde f^{\lambda,k}$, it holds that
\bea\label{eq_s4}
&\|\widetilde{Y}\|_\sF^{\lambda,k}\\
&\leq 2n(\sqrt{n}+1)\beta\sqrt{b}\|\widetilde{X}\|_\sQ^{\lambda,k}+n\|Y_\sV\|_\sS^{\lambda,k}+\frac{2n\beta\sqrt{2nb}}{\sqrt{\alpha}}\widetilde f^{\lambda,k}
\ena
where $\beta$ and  $b$ are given in Assumption \ref{assum} and Lemma \ref{lemma1}, respectively.\hfill$\square$
\end{lemma}
\begin{proof}[Proof of Theorem \ref{theo1}]
Note that $\lambda$ defined in Theorem \ref{theo1} satisfies all the conditions on $\lambda$ of the above four lemmas and a sufficient small $\bar{\gamma}$ will satisfy the condition of Lemma \ref{lemma_m3}. Thus, \eqref{eq_s1}, \eqref{eq_s2}, \eqref{eq_s3} and \eqref{eq_s4} hold. Let 
\bee
\ve(k) = [\|\widetilde{X}\|_\sQ^{\lambda,k},
\|Y_\sV\|_\sS^{\lambda,k},
\widetilde f^{\lambda,k},
\|\widetilde{Y}\|_\sF^{\lambda,k}]^\T, c=[c_1,c_2,c_3,0]^\T.
\ene  
Combining \eqref{eq_s1}, \eqref{eq_s2}, \eqref{eq_s3}, and \eqref{eq_s4}, we obtain that for all $k\in\bN$,
\bea\label{eq_lmi}
\ve(k) \preccurlyeq M\ve(k)+c
\ena
where $\preccurlyeq$ is the element-wise inequality and $M$ is a nonnegative matrix 
\bee
M=
\begin{bmatrix}
	0 & 0      & 0  &\frac{4\widetilde{n}\widetilde{t}\bar{\gamma}}{\theta^2} \\
	\frac{8\beta\widetilde n\sqrt{\widetilde n}\widetilde t}{\theta^2 \sqrt{n}} & 0      & 0   & \frac{8\beta\widetilde n\sqrt{\widetilde n}\widetilde t\bar{\gamma}}{\theta^2 n} \\
	\frac{48bn\beta}{\alpha\theta\sqrt{\bar{\gamma} \theta n}}         & \frac{16b}{\alpha\theta\sqrt{\bar{\gamma} \theta n}} & 0 & 0\\
	2n(\sqrt{n}+1)\beta\sqrt{b}  & n  &\frac{2n\beta\sqrt{2nb}}{\sqrt{\alpha}} &0
\end{bmatrix}
\ene

It follows from \eqref{eq_lmi} that if the spectral radius $\varrho(M)$ of $M$  is strictly less than 1, then $\|\widetilde{X}\|_\sQ^{\lambda,k}, \|Y_\sV\|_\sS^{\lambda,k}$ and $\widetilde f^{\lambda,k}$ are all bounded for all $k\in\bN$. 

Define a transformation matrix $T=\diag([1,1,\sqrt{\bar{\gamma}},\sqrt{\bar{\gamma}}])$. Then, we can choose a small $\bar{\gamma}$ to make $TMT^{-1}$ arbitrarily close to a strictly lower triangular matrix, and hence $\varrho(M)=\varrho(TMT^{-1})<1$ for a sufficiently small $\bar{\gamma}$ since the eigenvalues of a matrix are continuous functions on its elements. In fact, an upper bound of $\bar{\gamma}$ can be obtained by bounding $\|M^3\|_\infty$, which however can be conservative and thus we omit it here.

Define $\bar c:=\frac{\max\{c_1,c_2,c_3\}}{1-\varrho(M)}<\infty$, where $c_1,c_2$ and $c_3$ are given in \eqref{eq_c1}, \eqref{eq_c2} and \eqref{eq_c3}, respectively. It follows from \eqref{eq_lmi} that $\|\tilX\|_\sQ^{\lambda,k}\leq\bar c$ and $\widetilde f^{\lambda,k}\leq\bar c,\forall k$. We have
\bea
&\|\vx_i(k)-\text{Proj}_{\cX^\star}(\vx_{\pi}(k-1))\|_2\\
&\leq \|\vx_i(k)-\vx_{\pi}(k-1)\|_2+\|\vx_{\pi}(k-1)-\text{Proj}_{\cX^\star}(\vx_{\pi}(k-1))\|_2\\
&\leq\|\widetilde X(k)\|_\sQ+\sqrt{2/\alpha}\sqrt{f(\vec x_{\pi}(k-1))-f^\star}
\ena 
where the last inequality used the equivalence between the Polyak-\L ojasiewicz condition and the quadratic growth condition \cite[Theorem 2]{karimi2016linear}, i.e., $f(\vec x)-f^\star\geq\frac{\alpha}{2}\|\vx-\text{Proj}_{\cX^\star}(\vx)\|_2^2,\forall \vx\in\bR^m$. Let $\vx^\star(k)=\text{Proj}_{\cX^\star}(\vx_{\pi}(k-1))$ and $\|\vx_i-\vx^\star\|_2^{\lambda,k}$ be the $\lambda$-sequence of $\{\|\vx_i(k)-\vx^\star(k)\}\|_2$. We have
\bea
\|\vx_i-\vx^\star\|^{\lambda,k}\leq\|\tilX\|_\sQ^{\lambda,k}+\sqrt{2/\alpha}\widetilde f^{\lambda,k}\leq (1+\sqrt{2/\alpha}) \bar c,\ \forall i.
\ena
Combined with the boundedness of $\|\widetilde{Y}\|_\sF^{\lambda,k}$, the result in Theorem 1 follows by the definition of $\lambda$-sequence.
\end{proof}

\subsection{Two useful propositions}

We establish two important results in this subsection. The first one shows a property of $\lambda$-sequence and the second one recalls the contraction relation of gradient methods.

\begin{prop}\label{lemma4}
	Let $\{p(k)\},\{q(k)\}$ be nonnegative sequences satisfying
	\bea\label{eq1_lemma4}
	p(t+j)\leq rp(t)+\sum_{i=0}^{j-1}q(t+i)
	\ena
	where $r\in[0,1)$. If we choose $\lambda$ such that $\lambda^j\in(r,1)$, then the $\lambda$-sequences $p^{\lambda,k}$ and $q^{\lambda,k}$ in \eqref{eq1_sec4} satisfy
	\bee
	p^{\lambda,k}\leq\frac{j}{\lambda^j-r}q^{\lambda,k}+c_{\lambda},\ \forall k\in\bN,
	\ene
	where $c_{\lambda}=\frac{\lambda^j}{\lambda^j-r}\sum_{t=1}^{m}\lambda^{-t}p(t)$ is a constant. In particular, by letting $\lambda^j=\frac{2r}{1+r}$, we have
	\bea\label{eq_lemma4}
	&p^{\lambda,k}\leq\frac{2j}{r(1-r)}q^{\lambda,k}+c_{\lambda},\ \forall k\in\bN,
	\ena
\end{prop}
\begin{proof}
	It follows from \eqref{eq1_lemma4} that
	\bea\label{eq2_lemma3}
	&\lambda^{-(t+j)}p(t+j)\\
	&\leq \frac{r}{\lambda^j}\lambda^{-t}p(t)+\sum_{i=0}^{j-1}\frac{1}{\lambda^{j-i}}\lambda^{-(t+i)}q(t+i),\forall t=1,\cdots,k.
	\ena
	This gives $k$ inequalities by selecting $t=1,\cdots,k$. On the other hand, we have
	$
	\lambda^{-t}p(t)\leq \lambda^{-t}p(t),
	$
	which gives another $j$ inequalities by selecting $t=1,\cdots,j$. Take the maximum on both sides of these $k+j$ inequalities and use the definition of $\lambda$-sequence, we obtain
	\bea\label{eq4_lemma3}
	p^{\lambda,k+j}&\leq\frac{r}{\lambda^j}p^{\lambda,k}+\max\left\{q^{\lambda,k}\sum_{i=0}^{j-1}\frac{1}{\lambda^{j-i}},\max_{t=1,\cdots,j}\lambda^{-t}p(t)\right\}\\
	&\leq\frac{r}{\lambda^j}p^{\lambda,k+j}+q^{\lambda,k+j}\sum_{i=0}^{j-1}\frac{1}{\lambda^{j-i}}+\sum_{t=1}^{j}\lambda^{-t}p(t)
	\ena
	If $\lambda^j\in(r,1)$, then \eqref{eq4_lemma3} implies
	\bea
	p^{\lambda,k+j}&\leq\frac{\lambda^j\sum_{i=0}^{j-1}\frac{1}{\lambda^{j-i}}}{\lambda^j-r}q^{\lambda,k+j}+\frac{\lambda^j}{\lambda^j-r}\sum_{t=1}^{j}\lambda^{-t}p(t)\\
	&=\frac{1-\lambda^j}{(\lambda^j-r)(1-\lambda)}q^{\lambda,k+j}+\frac{\lambda^j}{\lambda^j-r}\sum_{t=1}^{j}\lambda^{-t}p(t)\\
	&\leq\frac{j}{\lambda^j-r}q^{\lambda,k+j}+\frac{\lambda^j}{\lambda^j-r}\sum_{t=1}^{j}\lambda^{-t}p(t)
	\ena
	for all $k+j\in\bN$, where we have used that $1-\lambda^j\leq j(1-\lambda)$. The result is obtained immediately.
	For $\lambda^j\in(\frac{2r}{1+r},1)$, we have $\lambda^j-r\geq r(1-r)/2$, and hence \eqref{eq_lemma4} follows.
\end{proof}

We introduce another important property of $\lambda$-sequence. For any nonnegative sequences $\{p(k)\}$, $\{q(k)\}$, letting $r(k)=p(k)+q(k)$, it holds that $r^{\lambda,k}\leq p^{\lambda,k}+q^{\lambda,k}$, which can be easily checked by definition.

The following proposition shows the convergence rate of a perturbed gradient descent method for minimizing functions satisfying the PL condition. As a special case, it recovers the linear convergence rate of the standard gradient descent method.

\begin{prop}\label{lemma6}
	Suppose that $f$ is $\beta$-Lipschitz smooth and satisfies the Polyak-\L ojasiewicz condition in Assumption \ref{assum}(c). Let $\eta\in(0,\frac{1}{2\beta})$, $\sigma =1-\alpha\eta(1-2\eta\beta)<1$, and $\vx^+=\vx-\eta\nabla f(\vx)+\varepsilon$. Then
	\bee
	f(\vx^+)-f^\star\leq\sigma(f(\vx)-f^\star)+\Big(\frac{2}{\eta}+\beta\Big)\|\varepsilon\|_2^2, \ \forall \vx\in\bR^n.
	\ene
\end{prop}
\begin{proof}
	It follows from the $\beta$-Lipschitz smoothness that
	\bea 
	&f(\vx^+)\\
	&\leq f(\vx)+\nabla f(\vx)^\T(-\eta\nabla f(\vx)+\varepsilon)+\frac{\beta}{2}\|-\eta\nabla f(\vx)+\varepsilon\|_2^2\\
	&\leq f(\vx)-\eta\|\nabla f(\vx)\|_2^2+\frac{\eta}{2}\|\nabla f(\vx)\|_2^2+\frac{2}{\eta}\|\varepsilon\|_2^2\\
	&\quad+\beta\eta^2\|\nabla f(\vx)\|_2^2+\beta\|\varepsilon\|_2^2\\
	&\leq f(\vx)-\eta(\frac{1}{2}-\eta\beta)\|\nabla f(\vx)\|_2^2+\Big(\frac{2}{\eta}+\beta\Big)\|\varepsilon\|_2^2
	\ena 
	Applying the Polyak-\L ojasiewicz inequality on $\|\nabla f(\vx)\|_2^2$ implies the desired result:
	\bea 
	&f(\vx^+)-f^\star\\
	&\leq f(\vx)-f^\star - 2\alpha\eta(\frac{1}{2}-\eta\beta)(f(\vx)-f^\star)+\Big(\frac{2}{\eta}+\beta\Big)\|\varepsilon\|_2^2\\
	&\leq \big(1-\alpha\eta(1-2\eta\beta)\big)(f(\vx)-f^\star)+\Big(\frac{2}{\eta}+\beta\Big)\|\varepsilon\|_2^2.
	\ena 
\end{proof}

We end this subsection with two inequalities which will be frequently used later. For any $A,B\in\bR^{n\times n}$,
\bea\label{ieq}
\|AB\|_\sF&\leq\|A\|_2\|B\|_\sF,\ \|A\|_2\leq\sqrt{\|A\|_\infty\|A\|_1},
\ena
and $\|A\|_2\leq\|A\|_\sF\leq\sqrt{n}$ for any row-stochastic matrix $A$.

\subsection{Proof of Lemma \ref{lemma_m1}}

Let $\widetilde t$ be defined as in Lemma \ref{lemma3}. It follows from \eqref{eq1_sec2} and \eqref{eq2_sec2} that
\bea\label{eq2_s1}
&\|\widetilde{X}(k+\widetilde t)\|_\sQ=\|Q(k+\widetilde t)\widetilde{X}(k+\widetilde t)\|_\sF\\
&\leq\|Q(k+\widetilde t)\Phi_{\widetilde{t}}^\sA (k)\widetilde{X}(k)\|_\sF\\
&\quad+\sum_{t=0}^{\widetilde t-1}\|Q(k+\widetilde{t})\Phi_{\widetilde{t}-t}^\sA (k+t)\Gamma I_{k+t}^a\widetilde{Y}(k+t)\|_\sF
\ena
By the stochastic vector $\phi_{\widetilde t}^\sA (k)$ in Lemma \ref{lemma3}, \eqref{eq2_s1} implies that
\bea\label{eq1_s1}
&\|\widetilde{X}(k+\widetilde t)\|_\sQ\\
&\leq\|Q(k+\widetilde t)(\Phi_{\widetilde{t}}^\sA (k)- \bone\phi_{\widetilde t-1}^\sA (k)^\T)Q(k)\widetilde{X}(k)\|_\sF\\
&\quad+\bar{\gamma}\sum_{t=0}^{\widetilde{t}-1}\|Q(k+\widetilde t)\|_2\|\Phi_{\widetilde{t}-t}^\sA (k+t)\|_2 \|I_{k+t}^a\widetilde{Y}(k+t)\|_\sF\\
&\leq\|Q(k+\widetilde t)\|_2\|\Phi_{\widetilde{t}}^\sA (k)- \bone\phi_{\widetilde t}^\sA (k)^\T\|_\sF\|Q(k)\widetilde{X}(k)\|_\sF\\
&\quad+2\bar{\gamma} \sum_{t=0}^{\widetilde{t}-1}\|\Phi_{\widetilde{t}-t}^\sA (k+t)\|_2\|\widetilde{Y}(k+t)\|_\sF\\
&\leq \sqrt{\widetilde n}\mu \|\widetilde{X}(k)\|_\sQ+2\bar{\gamma} \sqrt{\widetilde n}\sum_{t=0}^{\widetilde{t}-1}\|\widetilde{Y}(k+t)\|_\sF\\
\ena
where $\mu=\frac{\theta^2}{2\widetilde n}<1$ is given in Lemma \ref{lemma3}, the first inequality used the fact that $\|Q(k)\|_2\leq \sqrt{\widetilde n},\forall k$, $\Phi_{\widetilde{t}}^\sA (k)$ is row-stochastic and
\bea\label{eq4_s1}
&Q(k+t)(A-\bone\pi(k-1)^\T)Q(k)\\
&=Q(k+t)AQ(k)-Q(k+t)\bone\pi(k-1)^\T Q(k)\\
&=Q(k+t)(A-\bone_{\widetilde n}\pi(k-1)^\T)\\
&\quad-Q(k+t)(\bone_{\widetilde n}\pi(k-1)^\T-\bone_{\widetilde n}\pi(k-1)^\T)\\
&=Q(k+t)A-Q(k+t)\bone_{\widetilde n}\pi(k-1)^\T=Q(k+t)A
\ena
for any row-stochastic matrix $A$ and $k\in\bN$. The last inequality of \eqref{eq1_s1} follows from Lemma \ref{lemma3} and $\|\Phi_{\widetilde{t}}^\sA (k)\|_2\leq \sqrt{\widetilde n}$.

Note that $\mu<1/2$. In view of Proposition \ref{lemma4} and \eqref{eq1_s1}, we obtain that
\bee
\|\widetilde{X}\|_\sQ^{\lambda,k}\leq\frac{2\bar{\gamma}\sqrt{\widetilde n}\widetilde{t}}{\sqrt{\widetilde n}\mu(1-\sqrt{\widetilde n}\mu) }\|\widetilde{Y}\|_\sQ^{\lambda,k}+c_1
\ene
for any $\lambda^{\widetilde t}>\frac{2\sqrt{\widetilde n}\mu}{1+\sqrt{\widetilde n}\mu}$, where
\bee\label{eq_c1}
c_1=\frac{2\bar{\gamma}\sqrt{\widetilde n}\lambda^{\widetilde{t}}}{\lambda^{\widetilde{t}}-\sqrt{\widetilde n}\mu}\sum_{t=1}^{\widetilde{t}}\lambda^{-t}\|\widetilde{X}(t)\|_\sQ.
\ene
The desired result is obtained by setting $\mu=\frac{\theta^2}{2\widetilde n}$.
\qed

\subsection{Proof of Lemma \ref{lemma_m2}}

Let $\vv(k)=[v_1(k),\cdots,v_{\widetilde n}(k)]^\sT$ and $\cI_\sV(k)=\{i|v_i(k)=[V(k)]_{ii}=0\}$. Note that $v_i(k)\geq n\theta$ for all $i\in\cV,k\in\bN$ from Lemma \ref{lemma2}(b). It can be shown that the $i$-th row of $\tilY(k)$ is $\bzero_{m}^\T$ for all $i\in\cI_\sV(k)$, and thus $\tilY(k)=V(k)Y_\sV(k)$.

Let $R(k)=\nabla(k+1)-\nabla(k)$. It then follows from \eqref{eq1_sec2} that
\bee\label{eq2_s2}
Y_\sV(k+1)=\widetilde{B}_\sV(k)Y_\sV(k)+V(k+1)^{\dag}R(k)
\ene
where $\widetilde{B}_\sV(k)=V(k+1)^{\dag}\widetilde{B}(k)V(k)$ and one can prove that each row except the $i$-th ($i\in\cI_\sV(k)$) row of $\widetilde{B}_\sV(k)$ has row sum 1 using similar arguments as in Lemma 4 of \cite{nedic2016stochastic} and \cite{nedic2017achieving}. Using the definition of $\nabla(k)$ in \eqref{eq3_sec2} and Assumption \ref{assum}, we have
\bea\label{eq5_s2}
\|R(k)\|_\sF&= \|\nabla(k+1)-\nabla(k)\|_\sF\leq\beta\|{X}(k+1)-{X}(k)\|_\sF\\
&\leq\beta\|\widetilde{X}(k+1)-\widetilde{X}(k)\|_\sF
\ena
where we have used Assumption \ref{assum}.

Notice that
\bea\label{eq1_s2}
\|\widetilde{X}(k+1)-&\widetilde{X}(k)\|_\sF= \|\widetilde{A}(k)\widetilde{X}(k)-\Gamma I_k^a \widetilde{Y}(k)-\widetilde{X}(k)\|_\sF\\
&\leq \|(\widetilde{A}(k)-I)Q(k)\widetilde{X}(k)\|_\sF+\bar{\gamma}\|\widetilde{Y}(k)\|_\sF\\
&\leq2\sqrt{n}\|\widetilde{X}(k)\|_\sQ+\bar{\gamma}\|\widetilde{Y}(k)\|_\sF\\
\ena
where the second inequality follows from the row-stochasticity of $\widetilde{A}(k)$, and the last inequality is from that $\|\widetilde{A}(k)-I\|_2\leq\|\widetilde{A}(k)\|_2+\|I\|_2\leq \sqrt{\|\tilA(k)\|_1}+1\leq\sqrt{n}+1\leq 2\sqrt{n}$.
By Combining \eqref{eq5_s2} and \eqref{eq1_s2}, we obtain
\bee\label{eq3_s2}
\|R(k)\|_\sF\leq2\beta\sqrt{n}\|\widetilde{X}(k)\|_\sQ+\beta\bar{\gamma}\|\widetilde{Y}(k)\|_\sF
\ene

To analyze the sequence $\{Y_\sV(k)\}$, we define
\bea\label{eq6_s2}
\widetilde \Phi_t (k)&:=\widetilde B_\sV(k+t-1)\widetilde B_\sV(k+t-2)\cdots\widetilde B_\sV(k+1)\widetilde B_\sV(k)\\
&=V(k+t)^{\dag}\left(\prod_{l=t-1}^{1}\tilB(k+l)\tilI(k+l)\right)\tilB(k)V(k)\\
&=V(k+t)^{\dag}\Phi_t^\sB (k)V(k).
\ena
where
\bee
[\tilI(k)]_{ij}=[V(k)V(k)^{\dag}]_{ij}=\left\{
\begin{array}{ll}
	1, & \text{if $i=j,i\notin\cI_\sV(k)$}, \\
	0, & \text{otherwise.}
\end{array}\right.
\ene
and the last equality follows from that
$
\tilI(k+1)\tilB(k)V(k)=\tilB(k)V(k),\ \forall k\in\bN,
$
where we used the fact that $[\tilB(k)V(k)\bone_{\widetilde n}]=v_i(k+1)=0$ for any $i\in\cI_\sV(k+1)$, and thus the $i$-th row of $\tilB(k)V(k)$ is $\bzero_{m}^\T$.

We know from Lemma \ref{lemma2} and Lemma \ref{lemma3} that $\Phi_t^\sB (k)$ can be written as
\bee\label{eq8_s2}
\Phi_t^\sB (k) = \phi_{t}^\sB (k)\bone^\T + \Delta \Phi_{t}(k)
\ene
where $\|\Delta \Phi_{t}(k)\|_\sF\leq 2\rho^t$ and $\|\Delta \Phi_{\widetilde t}(k)\|_\sF\leq \mu<1$. Hence,
\bea
&\vv(k+t)=\Phi_{t}^\sB (k)\vv(k)=\phi_{t}^\sB (k)\bone^\T\vv(k)+\Delta \Phi_{t}(k)\vv(k)\\
&=\phi_{t}^\sB (k)\bone^\T\vv(0)+\Delta \Phi_{t}(k)\vv(k)=n\phi_{t}^\sB (k)+\Delta \Phi_{t}(k)\vv(k)
\ena
which implies that
\bee\label{eq7_s2}
\phi_{t}^\sB (k) = \frac{1}{n}\left(\vv(k+t)-\Delta \Phi_{t}(k)\vv(k)\right).
\ene
It then follows from \eqref{eq6_s2}, \eqref{eq8_s2} and \eqref{eq7_s2} that
\bea
\widetilde \Phi_t (k)&=\frac{1}{n}V(k+t)^{\dag}\vv(k+t)\bone_{\widetilde n}^\T V(k)\\
&\quad+V(k+t)^{\dag}\Delta \Phi_{t}(k)(I-\frac{1}{n}\vv(k)\bone_{\widetilde n}^\T)V(k)\\
&=\frac{1}{n}\widetilde \bone(k+t)\vv(k)^\T+C_{t}(k)
\ena
where $\widetilde \bone(k)=\tilI(k)\bone$, $C_{t}(k)=V(k+t)^{\dag}\Delta \Phi_{t}(k)(I-\frac{1}{\widetilde n}\vv(k)\bone_{\widetilde n}^\T)V(k)$ and
\bea
&\|C_{t}(k)\|_\sF\leq \|V(k+t)^{\dag}\|_2\|\Delta \Phi_{t}(k)\|_\sF\|(I-\frac{1}{ n}\vv(k)\bone_{\widetilde n}^\T)V(k)\|_\sF\\
&<\theta^{-1}\mu n.
\ena
where we used the fact that all entries of $V(k)^{\dag}$ are less than $\theta^{-1}$ by Lemma \ref{lemma2}. Thus
\bee
\left\|\widetilde \Phi_t (k)-\frac{1}{\widetilde n}\widetilde \bone(k+t)\vv(k)^\T\right\|_\sF\leq \theta^{-1}\mu {n}\leq\frac{\theta}{2}<\frac{1}{2}.
\ene

Now we turn to the sequence $\{Y_\sV(k)\}$. It follows from \eqref{eq1_c2} and \eqref{eq2_s2} that
\bea\label{eq9_s2}
&\|Y_\sV(k+\widetilde t)\|_\sS=\|S(k+\widetilde t)Y_\sV(k+\widetilde t)\|_\sF\\
&\leq\|S(k+\widetilde t)\widetilde \Phi_t (k)Y_\sV(k)\|_\sF\\
&\quad+\sum_{t=1}^{\widetilde{t}}\|S(k+t)\widetilde \Phi_{\widetilde t-t} (k+t)V(k+t)^{\dag}R(k+t-1)\|_\sF\\
&\leq\|S(k+\widetilde t)\widetilde \Phi_t (k)Y_\sV(k)\|_\sF\\
&\quad+\sum_{t=1}^{\widetilde{t}}\|S(k+t)V(k+t+1)^{\dag}\Phi_{\widetilde t-t}^\sB (k+t)R(k+t-1)\|_\sF
\ena
Similar to \eqref{eq4_s1}, we have
$
S(k+\widetilde t)\widetilde \Phi_t (k)=S(k+\widetilde t)(\widetilde \Phi_t (k)-\frac{1}{n}\widetilde \bone(k+t)\vv(k)^\T)S(k).
$
Following a similar argument as in \eqref{eq1_s1} and using \eqref{eq3_s2}, equation \eqref{eq9_s2} implies that
\bea
&\|Y_\sV(k+\widetilde t)\|_\sS\\
&\leq2\theta^{-1}\mu {n} \|Y_\sV(k)\|_\sS+2\sqrt{\widetilde n}\sum_{t=0}^{\widetilde{t}-1}\|V(k+t+2)^{\dag}R(k+t)\|_\sF\\
&\leq2\theta^{-1}\mu {n} \|Y_\sV(k)\|_\sS\\
&\quad+2\beta  \theta^{-1}\sqrt{\widetilde n}\sum_{t=0}^{\widetilde{t}-1}\left(2\sqrt{n}\|\widetilde{X}(k+t)\|_\sQ+\bar{\gamma}\|\widetilde{Y}(k+t)\|_\sF\right)\\
\ena
where we used the relation $\|V(k)^{\dag}\|_2\leq\theta^{-1},\forall k$ and \eqref{eq3_s2}.

By Proposition \ref{lemma4}, we have for any $\lambda^{\widetilde t}>(1+2\theta^{-1}\mu{n})/2$ that
\bea\label{eq10_s2}
&\|Y_\sV\|_\sS^{\lambda,k}\\
&\leq\frac{2\beta  \theta^{-1}n(1-\lambda^{\widetilde{t}})}{(\lambda^{\widetilde{t}}-2\theta^{-1}\mu{n} )(1-\lambda)}\|\widetilde{X}\|_\sQ^{\lambda,k}\\
&\quad+\frac{2\bar{\gamma}\beta \theta^{-1}\sqrt{n}(1-\lambda^{\widetilde{t}})}{(\lambda^{\widetilde{t}}-2\theta^{-1}\mu{n} )(1-\lambda)}\|\widetilde{Y}\|_\sF^{\lambda,k}+c'_2\\
&\leq\frac{2\beta \theta^{-1}\sqrt{\widetilde n}\widetilde t}{\theta^{-1}\mu{n}(1-2\theta^{-1}\mu{n}) }\left(\sqrt{n}\|\widetilde{X}\|_\sQ^{\lambda,k}+\bar{\gamma}\|\widetilde{Y}\|_\sF^{\lambda,k}\right)+c'_2
\ena
where
\bee\label{eq_c2}
c_2=\frac{2\beta  \theta^{-1}\sqrt{\widetilde n}\lambda^{\widetilde{t}}}{\lambda^{\widetilde{t}}-2\theta^{-1}\mu{n}}\sum_{t=1}^{\widetilde{t}}\lambda^{-t}\left(\sqrt{n}\|\widetilde{X}(t)\|_\sQ+\bar{\gamma}\|\widetilde{X}(t)\|_\sQ\right)
\ene
The desired result is obtained by setting $\mu=\frac{\theta^2}{2\widetilde n}$.
\qed

\subsection{Proof of Lemma \ref{lemma_m3}}

It follows from \eqref{eq1_sec2}, \eqref{abseq} and \eqref{eq9_s3} that
\bea\label{eq1_s3}
&\vec x_{\pi}(k+ b)=\vec \pi(k+b)^\T\widetilde X(k+ b)\\
&=\vec \pi(k+b)^\T\Phi_{ b}^\sA(k)\tilX(k)\\
&\quad- \vec \pi(k+b)^\T \sum_{t=0}^{ b-1}\Phi_{ b-t}^\sA(k+t+1)I_{k+t}^a\Gamma V(k+t)Y_\sV(k+t)\\
&=\vec \pi(k)^\T\tilX(k)-\frac{1}{n}\sum_{t=0}^{b-1}\eta_k(t)\bone_{\widetilde n}^\T \nabla(k+t)\\
&\quad-\sum_{t=0}^{b-1}\vr_k(t)(Y_\sV(k+t)-\frac{1}{n}\widetilde\bone(k+t)\bone_{\widetilde n}^\T \nabla(k+t))\\
\ena
where $V(k)$ and $\widetilde\bone(k)$ are defined in Lemma \ref{lemma_m2}, $\vr_k(t)=\vec \pi(k+b)^\T\Phi_{ b-t}^\sA(k+t+1)I_{k+t}^a\Gamma V(k+t)$ and $\eta_k(t)=\vr_k(t)\widetilde\bone(k+t)$. We have
\begin{subequations}\label{eq11_s3}
	\noeqref{eq11a_s3,eq11b_s3,eq11c_s3}
\begin{align}
\eta_k(t)&\leq \eta:={\bar{\gamma} n},\ \forall k,t\in\bN,\label{eq11a_s3}\\
\sum_{t=0}^{ b-1}\eta_k(t)&\geq  \vec \pi(k+b)^\T\sum_{t=0}^{b-1}\Phi_{ b-t}^\sA(k+t+1)I_{k+t}^a\Gamma\vv(k+t)\\
&\geq {\underline\gamma\theta n},\label{eq11b_s3}\\
\sum_{t=0}^{ b-1}\eta_k(t)&\leq b\eta\leq{\bar{\gamma}}nb<\frac{1}{4\beta},\ \forall k.\label{eq11c_s3}
\end{align}
\end{subequations}
where we used the relation $v_i(k)\geq n\theta,\forall i\in\cV,k\in\bN$ from Lemma \ref{lemma2}(b) and the fact that the sum of any row of $\sum_{t=0}^{b-1}\Phi_{ b-t}^\sA(k+t+1)I_{k+t}^a$ is not smaller than 1 to obtain \eqref{eq11b_s3}.

By introducing an auxiliary term $\sum_{t=0}^{ b-1}\eta_k(t)\nabla f(\vx_{\pi}(k))^\T$, Eq. \eqref{eq1_s3} becomes
\bea\label{eq8_s3}
&\vec x_{\pi}(k+ b)\\
&=\vec x_{\pi}(k)-\sum_{t=0}^{ b-1}\eta_k(t)\nabla f(\vx_{\pi}(k))^\T\\
&\quad+\underbrace{\sum_{t=0}^{ b-1}\eta_k(t)\left(\nabla f(\vx_{\pi}(k))^\T-\frac{1}{n}\bone_{\widetilde n}^\T \nabla(k+t)\right)}_{\vh(k)}\\
&\quad-\sum_{t=0}^{b-1}\vr_k(t)\left(Y_\sV(k+t)-\frac{1}{n}\widetilde\bone(k+t)\vv(k+t)^\T Y_\sV(k+t)\right)\\
\ena
where we have used the relation  $\bone_{\widetilde n}^\T\nabla(k)=\bone_{\widetilde n}^\T\tilY(k)=\vv(k)^\T Y_\sV(k)$ in \eqref{eq4_sec2}.

We now bound $\vh(k)$ in \eqref{eq8_s3}. Recall that $\nabla f(\vx)=\bone_n^\T\vec{\nabla f}(\bone_n\vx^\T),\forall \vx\in\bR^m$ and $\bone_{\widetilde n}^\T\nabla(k)=[\bone_{n};\bzero]^\T\nabla(k)=\bone_n^\T\nabla\vf(X(k)),\forall k$. Let $\widetilde\bone=[\bone_n;\bzero_{\widetilde{n}-n}]$, we have 
\begin{align}\label{eq4_s3}
&\|\vh(k)\|_\sF \\
&=\Big\|\sum_{t=0}^{ b-1}\eta_k(t)(\bone_n^\T\nabla \vf(\bone_n\vx_{\pi}(k))-\bone_{n}^\T\nabla\vf(X(k+t)))\Big\|_\sF\\
&\leq \sqrt{n}\beta\eta\sum_{t=0}^{ b-1}\|\bone_n\vx_{\pi}(k)-X(k+t)\|_\sF                                          \\
& \leq \sqrt{n}\beta\eta\Big(\sum_{t=0}^{ b-1}\bigg\|\widetilde\bone\vx_{\pi}(k)^\T-\widetilde\bone\pi_{k+t-1}^\T\tilX(k+t)\bigg\|_\sF\Big. +\\
&\qquad\Big.\sum_{t=0}^{ b-1}\bigg\|\widetilde\bone\pi_{k+t-1}^\T\tilX(k+t)-\begin{bmatrix}
I_n    & \bzero \\
\bzero & \bzero
\end{bmatrix}\tilX(k+t)\bigg\|_\sF\Big)    \\
& \leq \sqrt{n}\beta\eta\Big(\sum_{t=0}^{ b-1}\bigg\|\widetilde\bone({\pi}_{k-1}-\pi_{k+t-1})^\T Q(k+t)\tilX(k+t)\bigg\|_\sF\Big. \\
&\qquad+\Big.\sum_{t=0}^{ b-1}\bigg\|\begin{bmatrix}
I_n    & \bzero \\
\bzero & \bzero
\end{bmatrix}\Big(\bone_{\widetilde n}\pi_{k+t-1}^\T\tilX(k+t)-\tilX(k+t)\Big)\bigg\|_\sF\Big)                     \\
& \leq 3n\beta\eta\sum_{t=0}^{b-1}\|\tilX(k+t)\|_\sQ\leq 3\bar{\gamma} n^2\beta\sum_{t=0}^{b-1}\|\tilX(k+t)\|_\sQ
\end{align}
where we defined ${\pi}_{k}:=\pi(k)$ to save space.

By combining \eqref{eq8_s3} and \eqref{eq4_s3}, we obtain
\bea
\vec x_{\pi}(k+ b)&=\vec x_{\pi}(k)-\sum_{t=0}^{ b-1}\eta_k(t)\nabla f(\vx_{\pi}(k))^\T+\vh(k)\\
&\quad-\sum_{t=0}^{b-1}\vr_k(t)S(k+t)Y_\sV(k+t).
\ena
Then, Proposition \ref{lemma6} implies that 
\bea 
&f(\vec x_{\pi}(k+ b))-f^\star\leq \sigma^2(f(\vec x_{\pi}(k))-f^\star)\\
&+\big(\frac{2}{\sum_{t=0}^{ b-1}\eta_k(t)}+\beta\big)\Big\|\vh(k)-\sum_{t=0}^{b-1}\vr_k(t)S(k+t)Y_\sV(k+t)\Big\|_2^2
\ena 
where $\sigma=\sqrt{1-\alpha\sum_{t=1}^{ b}\eta_k(t)(1-2\beta\sum_{t=1}^{ b}\eta_k(t))}$, and we have $\sigma\leq 1-\frac{1}{2}\alpha \underline\gamma\theta n(1-2\beta\bar{\gamma} b n)<1-\frac{1}{4}{\alpha \underline\gamma\theta}{n}<1$ using $\sqrt{1-a}\leq 1-\frac{1}{2}a,\forall a\in(0,1)$, $\bar{\gamma} n\beta b\leq 0.25$, and \eqref{eq11_s3}. Using $\sqrt{a+b}\leq\sqrt{a}+\sqrt{b}$, we obtain
\bea\label{eq2_s3}
&\widetilde f(k+b)\\
&=\sqrt{f(\vec x_{\pi}(k+ b))-f^\star}\leq \sigma\sqrt{f(\vec x_{\pi}(k))-f^\star}\\
&\ +\Big(\sqrt{\beta}+\frac{\sqrt{2}}{\sqrt{\bar{\gamma}\theta n}}\Big)\Big(\|\vh(k)\|_\sF+\sum_{t=0}^{b-1}\|\vr_k(t)\|_2\|Y_\sV(k+t)\|_\sS\Big)\\
&\leq \sigma\widetilde f(k)+\Big(\bar{\gamma} n\sqrt{\beta}+\frac{\sqrt{2\bar{\gamma} n}}{\sqrt{\theta}}\Big)3n\beta\sum_{t=0}^{b-1}\|\tilX(k+t)\|_\sQ\\
&\quad+\Big(\bar{\gamma} n\sqrt{\beta}+\frac{\sqrt{2\bar{\gamma} n}}{\sqrt{\theta}}\Big)\sum_{t=0}^{b-1}\|Y_\sV(k+t)\|_\sS\\
&\leq \frac{2\sqrt{\bar{\gamma} n}}{\sqrt{\theta}}3n\beta\sum_{t=0}^{b-1}\|\tilX(k+t)\|_\sQ+\frac{2\sqrt{\bar{\gamma} n}}{\sqrt{\theta}}\sum_{t=0}^{b-1}\|Y_\sV(k+t)\|_\sS\\
&\quad+\sigma\widetilde f(k)
\ena 
where the first two inequalities follow from \eqref{eq11_s3}, \eqref{eq4_s3}, and $\|\vr_k(t)\|_2\leq \bar{\gamma} n,\forall k,t$. The last inequality follows from $\sqrt{\bar{\gamma} n\beta\theta}\leq\sqrt{\bar{\gamma} n\beta}\leq 0.5$.

Use the condition that $\lambda^b>1-\frac{1}{8}{\alpha \underline\gamma\theta}{n}$, $\sigma<1-\frac{1}{4}{\alpha \underline\gamma\theta}{n}$ and Proposition \ref{lemma4}, we obtain from \eqref{eq2_s3} that
\bea\label{eq3_s3}
\widetilde f^{\lambda,k}&\leq \frac{2\sqrt{\bar{\gamma} n}b\left({3n\beta}\|\tilX\|_\sQ^{\lambda,k}+\|Y_\sV\|_\sS^{\lambda,k}\right)}{(\lambda^b-(1-\frac{1}{4}{\alpha \bar{\gamma}\theta}{n}))\sqrt{\theta}}+c_3\\
&\leq\frac{16b}{\alpha\theta\sqrt{\bar{\gamma} \theta n}}\left({3n\beta}\|\tilX\|_\sQ^{\lambda,k}+\|Y_\sV\|_\sS^{\lambda,k}\right)+c_3\\
\ena
where
\bee\label{eq_c3}
c_3=\frac{16}{\alpha\theta\sqrt{\bar{\gamma} \theta n}}\sum_{t=0}^{b-1}\lambda^{-t}\left(3n\beta\|\widetilde{X}(t)\|_\sQ+\|{Y}_\sV(t)\|_\sS\right).
\ene
\quad

\subsection{Proof of Lemma \ref{lemma_m4}}
Since
\bea
&\widetilde{Y}(k)=V(k)Y_\sV(k)\\
&=V(k)S(k)Y_\sV(k)+\frac{1}{n}V(k)\widetilde \bone(k)\vv(k)^\T Y_\sV(k)
\ena
we have
\bea\label{eq2_s4}
\|\widetilde{Y}(k)\|_\sF&\leq\|V(k)\|_2\Big(\|Y_\sV(k)\|_\sS+\Big\|\frac{1}{n}\widetilde \bone(k)\vv(k)^\T Y_\sV(k)\Big\|_\sF\Big)\\
&\leq n\|Y_\sV(k)\|_\sS+\Big\|\widetilde \bone(k)\bone_{\widetilde n}^\T\nabla(k)\Big\|_\sF.
\ena

Note that
\begin{align}\label{eq3_s4}
&\|\widetilde\bone(k)\bone_{\widetilde n}^\T \nabla(k)\|_\sF=\Big\|\widetilde\bone(k)\big(\bone_n^\T \nabla\vf(X(k))-\bone_n^\T\nabla \vf(\bone_n(\vx^\star)^\T)\big)\Big\|_\sF\\
&\leq \Big\|\widetilde\bone(k)\bone_n^\T\Big\|_2 \Big\|\nabla\vf(X(k))-\nabla\vec{f}(\bone_n(\vx^\star)^\T)\Big\|_\sF\\
&\leq n\beta\sqrt{b+1}\|{X}(k)-\bone_n(\vx^\star)^\T\|_\sF\\
&= n\beta\sqrt{b+1}\Big\|\begin{bmatrix}
	I_n    & \bzero \\
	\bzero & \bzero
\end{bmatrix}\widetilde{X}(k)-\begin{bmatrix}
	\bone_n \\
	\bzero
\end{bmatrix}(\vx^\star)^\T\Big\|_\sF\\
&\leq 2n\beta\sqrt{b}\Big(\Big\|\begin{bmatrix}
	I_n    & \bzero \\
	\bzero & \bzero
\end{bmatrix}\widetilde{X}(k)-\begin{bmatrix}
	\bone_n \\
	\bzero
\end{bmatrix}\pi(k-1)^\T\widetilde{X}(k)\Big\|_\sF\\
&\quad+\Big\|\begin{bmatrix}
	\bone_n(\pi(k-1)^\T\widetilde{X}(k)-(\vx^\star)^\T) \\
	\bzero
\end{bmatrix}\Big\|_\sF\Big)\\
&\leq 2n\beta\sqrt{b}\Big\|(\begin{bmatrix}
	I_n    & \bzero \\
	\bzero & \bzero
\end{bmatrix}-\begin{bmatrix}
	\bone_n \\
	\bzero
\end{bmatrix}\pi(k-1)^\T) Q\widetilde{X}(k)\Big\|_\sF\\
&\quad+2n\beta\sqrt{b}\Big\|\begin{bmatrix}
\bone_n(\vec{x}_{\pi}(k)-(\vx^\star)^\T) \\
\bzero
\end{bmatrix}\Big\|_\sF\\
&\leq 2n\beta\sqrt{b}\Big((\sqrt{n}+1)\|\widetilde{X}(k)\|_\sQ+\sqrt{n}\|\vec{x}_{\pi}(k)-\vx^\star\|_\sF\Big)\\
&\leq 2n\beta\sqrt{b}\Big((\sqrt{n}+1)\|\widetilde{X}(k)\|_\sQ+\frac{\sqrt{2n}}{\sqrt{\alpha}}\sqrt{f(\vec x_{\pi}(k))-f^\star}\Big)
\end{align}
where $\vx^\star=\text{Proj}_{\cX^\star}(\vx_{\pi}(k))$ and the last inequality follows from that the Polyak-\L ojasiewicz condition implies the quadratic growth condition \cite[Theorem 2]{karimi2016linear}, i.e., $f(\vec x)-f^\star\geq\frac{\alpha}{2}\|\vx-\text{Proj}_{\cX^\star}(\vx)\|_2^2,\forall \vx$.

Substituting \eqref{eq3_s4} into \eqref{eq2_s4} yields
\bea
&\|\widetilde{Y}(k)\|_\sF\leq 2n(\sqrt{n}+1)\beta\sqrt{b}\|\widetilde{X}(k)\|_\sQ+n\|Y_\sV(k)\|_\sS\\
&+\frac{2n\beta\sqrt{2nb}}{\sqrt{\alpha}}\widetilde f(k).
\ena
The desired result is obtained by the definition of $\|\widetilde{Y}\|_\sF^{\lambda,k}$.
\qed

\section{Numerical Examples}\label{sec6}

We use APPG to train a multi-class logistic regression classifier in a distributed manner on the \emph{Covertype} dataset \cite{Dua2017UCI}, where the objective function takes the following form
\bee\label{lr}
f(X)= -\sum_{i=1}^{n_s}\sum_{j=1}^{n_c}l_j^i\log\left(\frac{\exp(\vec x_j^\T\vec s^i)}{\sum_{j'=1}^{n_c}\exp(\vec x_{j'}^\T\vec s^i)}\right)+\frac{\rho}{2}\|X\|_\sF^2.
\ene
Here $n_s=581012$ is the number of training instances, $n_c=7$ is the number of classes, $n_f=55$ is the number of features, $\vec s^i\in\bR^{55}$ is the feature vector of the $i$-th instance, $\vec{l}^i=[l_1^i,...,l_7^i]^\T$ is the label vector of the $i$-th instance using the one-hot encoding, $X=[\vec x_1,...,\vec x_7]\in\bR^{n_f\times n_c}$ is the parameters to be optimized, $\rho=20$ is a regularization factor.

{\bf Environment: }APPG is implemented in Python 3.6 with OpenMPI 1.10 on Ubuntu 14.04. The hardware is a server with 28 Xeon E5-2660 cores. Each core serves as a computing node.

{\bf Distributed Data: }We first normalize non-categorical features by subtracting the mean and dividing by the standard deviation in the whole dataset. Then, we \emph{sort} the data by digit label, and sequentially partition it into $n$ parts (with different sizes), where each node (core) only has \emph{exclusive} access to one part. Thus, we are dealing with distributed datasets. 

{\bf Topology: }The directed network among nodes is as follows: Each node $i$ sends messages to node $\text{mod}(2^j+i,n)$, where $j\in\bN\cap[0,\log_2(n))$ and $\text{mod}(a,b)$ returns the remainder after division of $a$ by $b$. Thus, each node has $\cO(\log(n))$ out-neighbors, which results in a relatively sparse directed networks. We also implement APPG over other networks in Section \ref{sec6b}. Note that gossip-based asynchronous algorithms generally cannot work over these directed networks.

{\bf Stepsize: } The stepsize of each algorithm is tuned via a grid search around $0.5/n_s$.

{\bf Local Termination Criteria: }Node $i$ stops locally if the value of $\vy_i$ in last $n$ consecutive iterations are less than $300/n_s$.

\subsection{Convergence performance and linear speedup}\label{sec6a}

We implement APPG over $n=1,6,12,18,24$ nodes ($n=1$ is used as a baseline since APPG reduces to standard centralized gradient descent method). The training loss w.r.t. running time is plotted in Fig. \ref{fig1a}, which validates the convergence of APPG, and shows that the training time is significantly reduced with the increase of number of nodes.

Fig. \ref{fig1c} plots the training loss of a synchronous version of APPG, which is done by adding a barrier after each update (c.f Section \ref{sec3.3}).  The result shows that its convergence rate is slower than APPG for $n>1$.

Fig. \ref{fig1b} depicts the training loss w.r.t. the number of iterations. We find that the number of iterations required to achieve the same accuracy is close to each other for different number of nodes. The time for a node to finish an iteration is proportional to the size of its local dataset, and hence is roughly inversely proportional to the number of nodes, which suggests that using $n$ nodes may reduce $\cO(n)$ times of training time than that of one node.

To further illustrate this property, we study the speedup of APPG defined as $S_n:=T_n/T_1$, where $T_n$ is the running time of the APPG with $n$ node(s) when the training loss decays to $0.005$. Fig. \ref{fig2} shows that the APPG achieves a roughly linear speedup in convergence rate w.r.t. the number of nodes. One can also find that the synchronous version of APPG has an approximately linear speedup when the number of cores is small, but it decreases fast when the number of cores is relatively large. 

Ideally, the speedup would be $n$ when using $n$ nodes. However, the communication among nodes introduces delays and staleness to the algorithm, which degrades the convergence rate. In practice, a higher speedup than Fig. \ref{fig2} can be achieved by using a low-latency network with larger bandwidth.

\begin{figure*}[!t]
	\centering
	\begin{subfigure}[c]{0.3\linewidth}
		\centering
        \includegraphics[width=\linewidth]{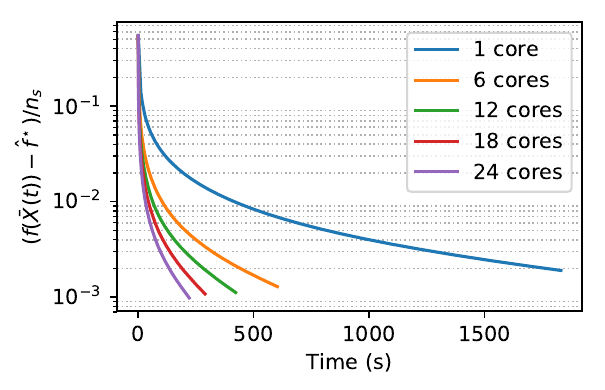}
        \caption{}\label{fig1a}
	\end{subfigure}
	\begin{subfigure}[c]{0.3\linewidth}
		\centering
        \includegraphics[width=\linewidth]{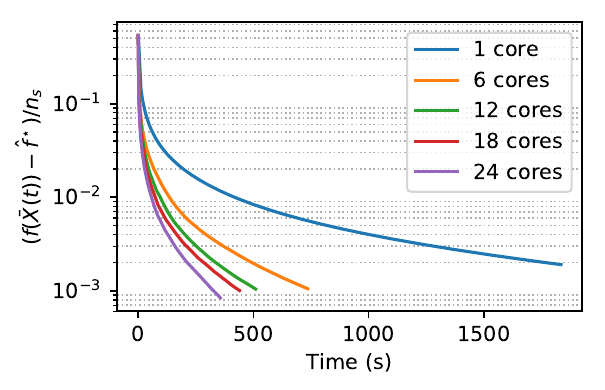}
        \caption{}\label{fig1c}
	\end{subfigure}
	\begin{subfigure}[c]{0.3\linewidth}
		\centering
        \includegraphics[width=\linewidth]{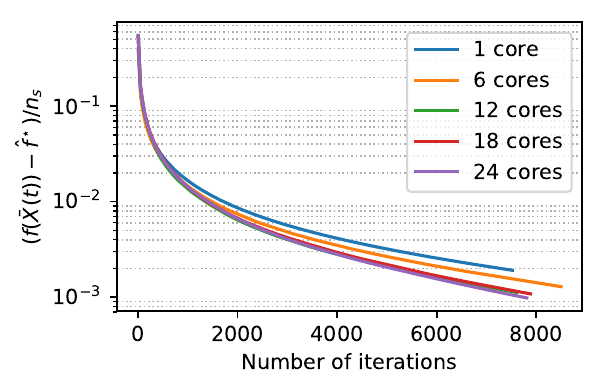}
        \caption{}\label{fig1b}
	\end{subfigure}
	\caption{Convergence performance with different number of nodes. (a) Training loss w.r.t. running time of APPG. (b) Training loss w.r.t. running time of the synchronous version of APPG. (c) Training loss w.r.t. number of iterations (epochs) of APPG.}
	\label{fig1}
\end{figure*}

\begin{figure}[!t]
	\centering
	\begin{subfigure}[c]{0.47\linewidth}
		\centering
        \includegraphics[width=\linewidth]{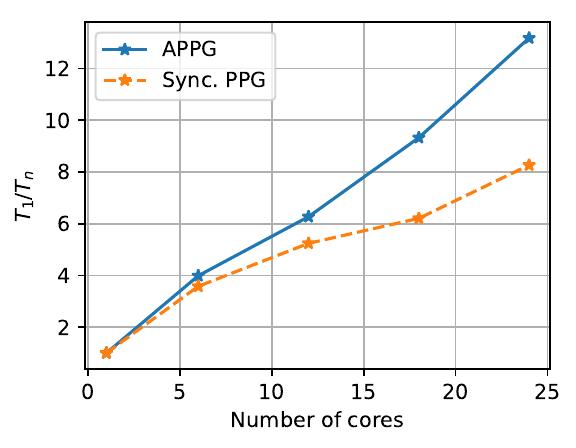}
        \caption{}\label{fig2}
	\end{subfigure}
	\begin{subfigure}[c]{0.5\linewidth}
		\centering
        \includegraphics[width=\linewidth]{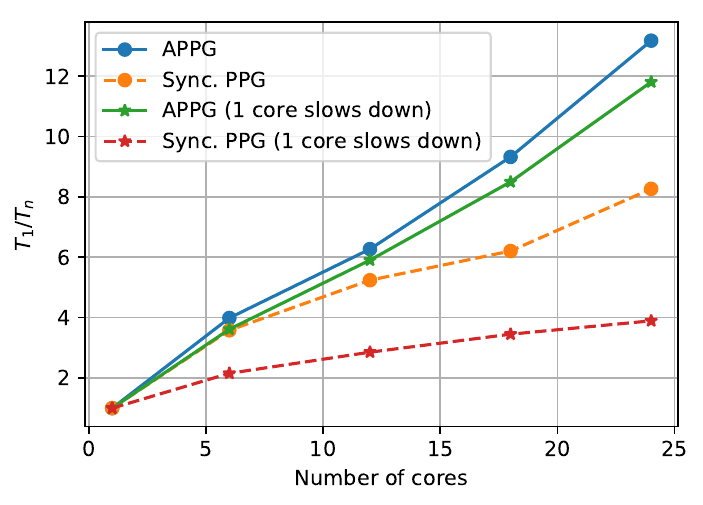}
        \caption{}\label{fig5}
	\end{subfigure}
	\caption{(a) Speedup in running time of APPG and the synchronous implementation of APPG w.r.t. the number of cores. $T_n$ is the running time of the APPG with $n$ core(s) when the training loss decays to $0.005$. (b) Speedup of APPG and the `synchronized' APPG when one core slows down.}
\end{figure}

\begin{figure}[!t]
	\centering
	\begin{subfigure}[c]{0.48\linewidth}
		\centering
        \includegraphics[width=\linewidth]{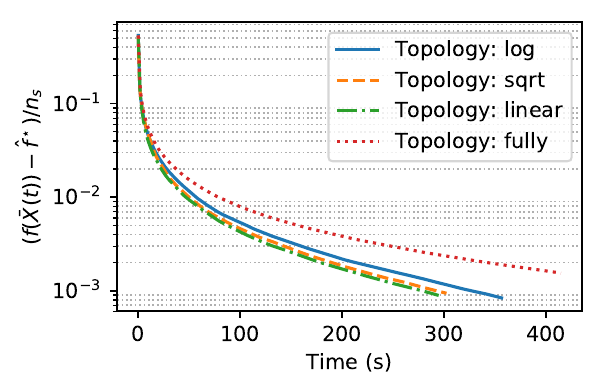}
        \caption{}\label{fig3a}
	\end{subfigure}
	\begin{subfigure}[c]{0.48\linewidth}
		\centering
        \includegraphics[width=\linewidth]{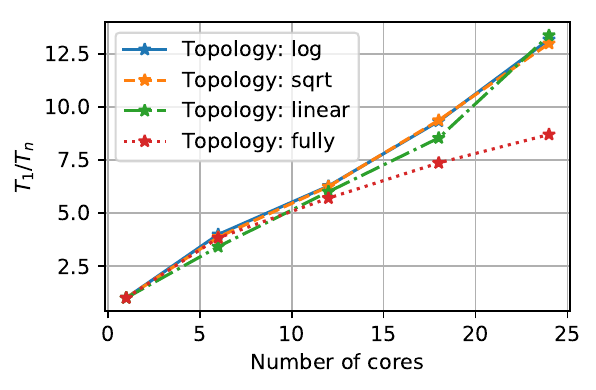}
        \caption{}\label{fig3b}
	\end{subfigure}
	\caption{(a) Convergence rate of APPG using 24 cores over different topologies. (b) Speedup in running time w.r.t. the number of cores over different topologies.}
	\label{fig3}
\end{figure}

\subsection{Effect of network topology}\label{sec6b}

The communication topology may largely affect the convergence performance of distributed algorithms, which is empirically studied in this subsection, where we test the APPG under the following directed graphs. 
\begin{enumerate}[label=(\alph*)]
	\item \texttt{log} topology (default): Node $i$ sends messages to node $\text{mod}(2^j+i,n)$, where $j\in\bN\cap[0,\log_2(n))$.
	\item \texttt{sqrt} topology: Node $i$ sends messages to node $\text{mod}(j^2+i+1,n)$, where $j\in\bN\cap[0,\sqrt{n})$.
	\item \texttt{linear} topology: Node $i$ sends messages to node $\text{mod}(5j+i+1,n)$, where $j\in\bN\cap[0,n/5)$.
	\item \texttt{fully} topology: Fully connected graph, a node sends information to all the rest nodes.
\end{enumerate}
The \texttt{log} topology has the  sparsest edges while the \texttt{fully} topology is the densest one.

Fig. \ref{fig3a} shows the convergence rate in running time of 24 cores over these topologies, and Fig. \ref{fig3b} depicts the speedup. For  \texttt{log}, \texttt{sqrt} or \texttt{linear} topologies, the convergence rate is slightly faster if the graph is denser, which is because a denser graph accelerates the information mixing speed. However, there is a sharp reduction in convergence rate  when the graph is too dense as in the \texttt{fully} topology. The reason is such a dense graph results in large amount of transmitted data per iteration, which heavily increases the communication overhead and the staleness in gradient computation. In practice, an appropriate topology should be designed according to the network bandwidth and latency.

\subsection{Robustness of APPG to slow cores}\label{sec6d}

We evaluate the robustness of APPG by forcing one core in the network to slow down. This is achieved by adding an artificial waiting time (20ms, a normal iteration takes about 15ms with 24 cores) after each local iteration of a node, which simulates either the slow computation or slow communication.

Fig. \ref{fig5} shows the speedup of the APPG and the synchronous implementation of APPG in this scenario. It indicates that the synchronous counterpart of APPG has a sharp reduction in convergence rate even when only 1 core slows down. In contrast, APPG still keeps an almost linear speedup. This result is also consistent with that in \cite{lian2018asynchronous,zhang2018asyspa}.

\begin{figure}[!t]
	\centering
	\begin{subfigure}[c]{0.49\linewidth}
		\centering
        \includegraphics[width=\linewidth]{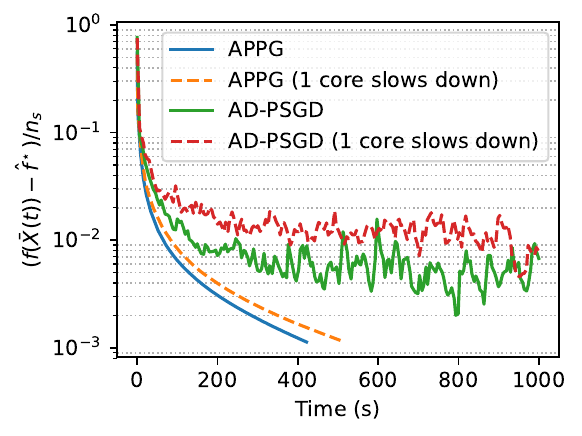}
        \caption{12 cores}\label{fig6a}
	\end{subfigure}
	\begin{subfigure}[c]{0.49\linewidth}
		\centering
        \includegraphics[width=\linewidth]{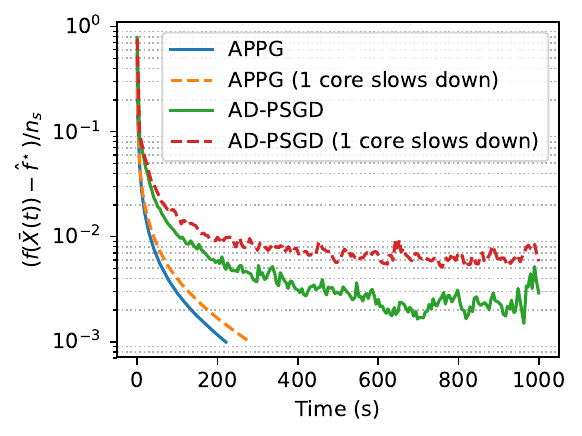}
        \caption{24 cores}\label{fig6b}
	\end{subfigure}
	\caption{Convergence of APPG and AD-PSGD with full local gradient when one node is artificially slowed down by adding 20ms waiting time after each iteration.}
	\label{fig6}
\end{figure}

Introducing the slowing core also brings an easily overlooked problem of asynchronous algorithms, that is, the cores have \emph{uneven} update rates. To show its effect to the performance, we compare the proposed algorithm to a gossip-based asynchronous algorithm AD-PSGD \cite{lian2018asynchronous} with full local gradients. Note that APPG can only work over \emph{undirected} networks, and hence we modify the network for it by adding a reversed edge to each edge in the directed network, while APPG is still implemented over the directed network. Fig. \ref{fig6} shows the result over 12 nodes and 24 nodes, where the AD-PSGD fails to converge to the exact optimum. In contrast, APPG converges exactly despite that the convergence rate is reduced a bit.

\section{Conclusion}\label{sec7}
This paper has proposed a fully asynchronous algorithm (APPG) for distributed optimization. It allows nodes to connect via a directed communication network and update with uncoordinated computation and stale information from neighbors. Linear convergence rate of APPG is achieved for possibly non-convex  Lipschitz-smooth functions satisfying the PL condition. The performance of APPG is also demonstrated via a logistic regression problem. Future works may focus on accelerating APPG and extending it to stochastic optimization settings.

\bibliographystyle{IEEEtran}
\bibliography{mybibf}

\end{document}